\documentclass[12pt,onecolumn,oneside,a4paper,peerreview]{IEEEtran_kit}
\usepackage{cite}
\usepackage{float}
\usepackage{graphicx}
\usepackage{amsmath}
\usepackage{times}
\usepackage{graphicx}
\usepackage{latexsym}
\usepackage{graphicx}
\usepackage{bm}
\usepackage{amssymb}
\usepackage[center]{caption2}
\usepackage{stfloats}
\usepackage{cases}
\usepackage{array}
\usepackage{setspace}
\usepackage{fancyhdr}
\usepackage{color}
\usepackage{subfigure}
\usepackage{citesort}
\usepackage{multirow}
\usepackage{amsmath}
\usepackage{cases}
\usepackage[noend]{algpseudocode}
\usepackage{algorithmicx,algorithm}
\def\bm#1{\mbox{\boldmath $#1$}}

\newtheorem{theorem}{Theorem}
\newtheorem{lemma}{Lemma}

\newtheorem{corollary}{Corollary}

\newtheorem{proposition}{Proposition}
\newtheorem{remark}{Remark}
\renewcommand{\baselinestretch}{0.97}

\begin{document}
\title{Location Information Aided Multiple Intelligent Reflecting Surface Systems}
 \author{\authorblockN{Xiaoling Hu, Caijun Zhong, Yu Zhang, Xiaoming Chen, and Zhaoyang Zhang
 \thanks{X. Hu, C. Zhong, X. Chen and Z. Zhang are with the College of Information Science and Electronic Engineering, Zhejiang University, Hangzhou, China (Email:  \{11631052, caijunzhong, chen\_xiaoming, ning\_ming\}@zju.edu.cn).}
  \thanks{Y. Zhang is College of Information Engineering, Zhejiang University of Technology, Hangzhou 310014, China (Email: yzhang@zjut.edu.cn).}
 }}
\maketitle
\begin{abstract}
This paper proposes a novel location information aided multiple intelligent reflecting surface (IRS) systems. Assuming imperfect user location information, the effective angles from the IRS to the users are estimated, which is then used to design the transmit beam and IRS beam. Furthermore, closed-form expressions for the achievable rate are derived. The analytical findings indicate that the achievable rate can be  improved by increasing the number of base station (BS) antennas or reflecting elements. Specifically, a power gain of order $N M^2$ is achieved, where $N$ is the antenna number and $M$ is the number of reflecting elements. Moreover, with a large number of reflecting elements, the individual signal to interference plus noise ratio (SINR) is proportional to $M$, while becomes proportional to $M^2$ as non-line-of-sight (NLOS) paths vanish.
Also, it has been shown that high location uncertainty would significantly degrade the achievable rate.
Besides, IRSs  should  be  deployed  at distinct  directions (relative to the BS) and be far away from each other to reduce the interference from multiple IRSs. Finally, an optimal power allocation scheme has been proposed to improve the system performance.
\end{abstract}

\newpage

\section{Introduction}
With the desirable feature of low energy consumption and low hardware complexity, the intelligent reflecting surface (IRS) has recently emerged as a key candidate technology for future wireless communication systems \cite{tang2019mimo}. Specifically, an IRS is a meta-surface comprising a large number of low-cost, nearly passive, reflecting elements, each of which can independently reflect the incident signal with adjustable phase shifts.
By properly adjusting its phase shifts, the IRS can enhance the desired signal power at the receiver, thereby helping combat the unfavorable wireless propagation environment.

Due to the aforementioned benefits, IRS-aided wireless communications have attracted considerable research interests. In \cite{wu2018intelligent,wu2019intelligent,J.Gao,X.Hu,Yu2019MISO}, the joint transmit beamforming and IRS beamforming problem was studied, while in \cite{huang2019reconfigurable}, transmit power allocation and phase shift beamforming were jointly optimized. Later on, the works \cite{wu2019beamforming_TCOM} and \cite{wu2019beamforming} investigated a more realistic case where the IRS only has a finite number of phase shifts at each element. The multi-user multi-IRS scenario has been considered in \cite{li2019joint}, where joint active and passive beamformers are designed based on the assumption that global channel state information (CSI) is available. Moreover, the integration of IRS with other promising technologies has been studied, including cognitive radio \cite{yuan2019intelligent}, non-orthogonal multiple
access (NOMA) \cite{yang2019intelligent,fu2019intelligent,mu2019exploiting,ding2019simple}, two-way communications \cite{Y.Zhang}, millimeter wave (mmWave)\cite{wang2019joint,wang2019intelligent}, terahertz communications \cite{ning2019channel}, physical-layer security \cite{cui2019secure,Chen2019Intelligent,yu2019enabling} and simultaneous wireless information and power transfer (SWIPT)\cite{wu2019joint,pan2019intelligent}.

However, there are two main challenges when performing IRS beamforming in practice. The first one is the requirement of instantaneous CSI, as adopted by all the aforementioned works. Due to the passive architecture of the IRS and large number of reflecting elements, CSI acquisition is highly non-trivial. The conventional CSI estimation techniques are not applicable, and the typical method is to estimate the cascaded channel instead of individual channels \cite{mishra2019channel,zheng2019intelligent,he2019cascaded,wang2019channel,you2019intelligent}. However, with a large number of reflecting elements, the channel training overhead becomes prohibitively high. To reduce the training overhead, the works \cite{yang2019intelligent,zheng2019intelligent} proposed to divide the IRS elements into groups where only the effective channel for all elements in each group are estimated. But this method comes at the cost of degraded IRS passive beamforming performance, since the phase shifts of reflecting elements in each group needs to be set identical.

Another challenge is that a separate link is required for information exchange between the base station (BS) and the IRS. In general, phase shifts are designed at the BS according to the instantaneous CSI, and then shared with the IRS via the separate link \cite{abeywickrama2019intelligent,zheng2019intelligent}. Due to the time varying nature of the wireless channel, phase shifts designed with instantaneous CSI need to be updated frequently. In addition, as the number of reflecting elements becomes larger, the required capacity of the separate link also increases. To reduce the frequency of information exchange between the BS and IRS, the work \cite{han2019large,X.Hu2} proposed to use statistical CSI for phase shift design.

 To address the aforementioned two challenges, we propose a novel location information aided multi-IRS design framework, where the design of the transmit  beam and phase shifts only relies on statistical CSI obtained from location information.
Compared with the traditional design framework which usually requires full CSI and involves large training overhead,   our proposed design framework has the following three major advantages. First of all, the location information can be easily obtained by global position system (GPS), hence the training overhead can be substantially reduced. Secondly, the location information varies much slower compared with the instantaneous CSI, hence no frequent update is required. Thirdly, only very small amount of location information needs to be shared between the BS and IRS, hence only low-capacity connection is required, which further reduces the hardware implementation cost.
The main contributions of this paper are summarized as follows:
\begin{itemize}
    \item Utilizing the imperfect location information, the angle information of the line-of-sight (LOS) path is estimated, and the statistical properties of the estimation error are characterized.
    \item  Based on the estimated angle information, a low-complexity beamforming scheme is proposed, where closed-form expressions of the transmit beam and phase shift beam are obtained.
Simulation results show that when the individual user rate requirement is low, the proposed
beamforming scheme is superior to the joint optimization scheme in   \cite{wu2019intelligent}.
    \item The closed-form expression for the achievable rate of the proposed scheme is presented, which facilitates the study of the impact of key system parameters such as location accuracy, number of reflecting elements, number of antennas and Rician K-factor.
    \item Finally,  under the proposed beamforming scheme, an optimal power control scheme is proposed to minimize the total transmit power, subject to individual user rate constraints.
\end{itemize}

The remainder of the paper is organized as follows. In Section \ref{s1}, we introduce the multi-IRS system, while in Section \ref{s2}, we propose an angle estimation scheme based on location information. According to the estimated  effective angles, a low-complexity scheme for BS beamforming and IRS beamforming is presented in Section \ref{s3}. Then, we give a detailed analysis of the achievable rate in Section \ref{s4}. Also, an optimal power control scheme is proposed in \ref{PC}.
Numerical results and discussions are provided in Section \ref{s5}, and finally Section \ref{s6} concludes the paper.

Notation: Boldface lower case and upper case letters are used for column vectors and matrices, respectively. The superscripts ${\left(\cdot\right)}^{*}$, ${\left(\cdot\right)}^{T}$, ${\left(\cdot\right)}^{H}$, and ${\left(\cdot\right)}^{-1}$ stand for the conjugate, transpose, conjugate-transpose, and matrix inverse, respectively. Also, the Euclidean norm, absolute value, Hadamard product  are denoted by $\left\| \cdot \right\|$, $\left|\cdot\right|$ and $\odot$ respectively. In addition, $\mathbb{E}\left\{\cdot\right\}$ is the expectation operator, and $\text{tr}\left(\cdot\right)$ represents the trace.
For a matrix ${\bf A}$, ${[\bf A]}_{mn}$ denotes its entry in the $m$-th row and $n$-th column, while for a vector ${\bf a}$, ${[\bf a]}_{m}$ denotes the $m$-th entry of it. Besides, $j$ in $e^{j \theta}$ denotes the imaginary unit.
Finally, $z \sim \mathcal{CN}(0,{\sigma}^{2})$ denotes a circularly symmetric complex Gaussian random variable (RV) $z$ with zero mean and variance $\sigma^2$, and $z \sim \mathcal{N}(0,{\sigma}^{2})$ denotes a real valued Gaussian RV.

\section{System Model} \label{s1}
We consider a multi-IRS system, as illustrated in Fig. \ref{f0}, where one BS with $N$ antennas communicates with $K$ single antenna users, each of which is assisted by an IRS with $M$ reflecting elements.{ \footnote{
  When the number of users is large, it is infeasible to associate each user with a unique IRS. In such scenario, the IRSs can be assigned to users with poor coverage. Alternatively, effective user scheduling protocol can be designed to exploit the benefit of IRSs.}} We further assume that the direct links between the BS and users do not exist  due to blockage or unfavorable propagation environments \cite{huang2019reconfigurable}.
The BS is connected with IRSs via low-capacity hardware links so that they can exchange information (e.g. CSI and phase shifts).
As in \cite{yuan2019intelligent,he2019adaptive}, we assume that both the BS and IRSs use uniform linear arrays (ULA).{ \footnote{ The main reason to use ULA instead of uniform planar array (UPA) is that the simple structure of the ULA  helps to get some useful and intuitive insights. In addition, the UPA can be regarded as multiple ULAs. As such, the analytical results obtained for ULA can be readily extended to the UPA structure.}}
Without loss of generality, we assume that the BS is located at the origin of a Cartesian coordinate system, and the ULAs of both the BS and IRSs are along the $y$ axis.

\begin{figure}[!ht]
  \centering
  \includegraphics[width=3.5in]{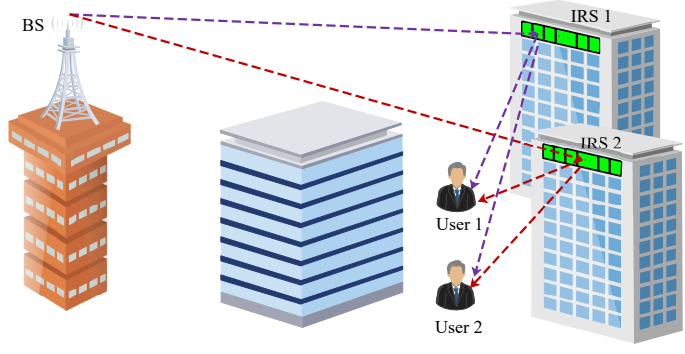}
  \caption{Model of the location information aided multi-IRS system with $K=2$ users assisted by 2 IRSs coated on two buildings. The reflecting elements are represented by green rectangles.}
  \label{f0}
\end{figure}

To fully exploit the potential of IRSs, they are usually deployed in desirable locations with LOS paths to both the BS and users. As such, we use angle-domain Rician fading to model the channels between IRSs and the BS or users\cite{wu2019beamforming_TCOM}.  Specifically, the channel from the BS to the $m$-th $\text{IRS}$ can be expressed as
\begin{align} 
{\bf G}_{\text{B2I},m}\!=\!\sqrt{\alpha_{\text{B2I},m}\frac{v_{\text{B2I},m}}{v_{\text{B2I},m}\!+\!1}} \bar{\bf G}_{\text{B2I},m} \!+\!\sqrt{\alpha_{\text{B2I},m}\frac{1}{v_{\text{B2I},m}\!+\!1}} \widetilde{\bf G}_{\text{B2I},m},
\end{align}
where $\widetilde{\bf G}_{\text{B2I},m}{  \in \mathbb{C}^{M \times N} } $ is the non-line-of-sight (NLOS) component, whose elements follow the $\mathcal{CN}\left(0,1\right)$ distribution, $\alpha_{\text{B2I},m}$ models large-scale fading, and $v_{\text{B2I},m} $ is the Rician K-factor.
The LOS component $\bar{\bf G}_{\text{B2I},m} {  \in \mathbb{C}^{M \times N} } $ is given by
\begin{align}
 \bar{\bf G
 }_{\text{B2I},m}= {\bf b} \left(\vartheta_{\text{y-B2Ia},m} \right)
{\bf a}^T \left(\vartheta_{\text{y-B2I},m} \right),
\end{align}
where ${\bf a}\left(\vartheta_{\text{y-B2I},m} \right) {  \in \mathbb{C}^{N \times 1} } $ is the array response vector at the BS, with the effective angle  of departure (AOD), i.e., the phase difference between two adjacent antennas along the $y$ axis,  given by
 \begin{align}
 \vartheta_{\text{y-B2I},m}=-\frac{2 \pi d_{\text{BS}}}{\lambda} \cos \theta_{\text{B2I},m} \sin \phi_{\text{B2I},m},
  \end{align}
 where $d_{\text{BS}}$ is the  distance between two adjacent BS antennas,
$\lambda$ is the carrier wavelength, $ \theta_{\text{B2I},m}$ and $\phi_{\text{B2I},m}$ are the elevation and azimuth  AODs from the BS to the $m$-th IRS, respectively.

Similarly, ${\bf b} \left(\vartheta_{\text{y-B2Ia},m} \right) {  \in \mathbb{C}^{M \times 1} }  $ is the array response vector at the $m$-th IRS, where the effective angle of arrival (AOA) is given by
\begin{align}
 \vartheta_{\text{y-B2Ia},m}= \frac{2 \pi d_{\text{IRS}} }{\lambda} \cos \theta_{\text{B2Ia},m} \sin \phi_{\text{B2Ia},m},
 \end{align}
 where $d_{\text{IRS}}$ is the  distance between two adjacent reflecting elements,
 $ \theta_{\text{B2Ia},m}$ and $\phi_{\text{B2Ia},m}$ are the elevation and azimuth AOAs at the $m$-th IRS, respectively. Furthermore, without loss of generality, we assume $d_{\text{BS}}=d_{\text{IRS}}=\frac{\lambda}{2}$.

 Specifically,  the $n$-th element of ${\bf a} \left(\vartheta  \right)$ and the $l$-th element of ${\bf b} \left( \vartheta  \right)$ are given by
\begin{align}
&a_n=e^{j\pi
\left(n-1\right) \vartheta }, n=1,...,N,\\
&b_{l}=e^{ j\pi\left(l-1\right)
 \vartheta   }, l=1,...,M.
\end{align}

Similarly, the channel from the $m$-th IRS to the $k$-th user is given by ${\bf g}_{\text{I2U},mk}^T {  \in \mathbb{C}^{1 \times M} } $:
\begin{align}
{\bf g}_{\text{I2U},mk}^T=\sqrt{\alpha_{\text{I2U},mk}\frac{v_{\text{I2U},mk}}{v_{\text{I2U},mk}+1}} \bar{\bf g}_{\text{I2U},mk}^T
+\sqrt{\alpha_{\text{I2U},mk}\frac{1}{v_{\text{I2U},mk}+1}} \widetilde{\bf g}_{\text{I2U},mk}^T    ,
\end{align}
with the LOS component $\bar{\bf g}_{\text{I2U},mk}^T {  \in \mathbb{C}^{1 \times M} }  $ given by
\begin{align}
   \bar{\bf g}_{\text{I2U},mk}^T= {\bf b}^T \left(\vartheta_{\text{y-I2U},mk} \right).
\end{align}


During the downlink data transmission phase, the BS broadcasts the signal
\begin{align}
{\bf x}=\sum\limits_{i=1}^{K} {\bf w}_i s_i,
\end{align}
where ${\bf w}_i {  \in \mathbb{C}^{N \times 1} } $ is the transmit beamforming vector and $s_i$ is the symbol for the $i$-th user, satisfying $\mathbb{E}\left\{ | s_i|^2\right\}=1$.

Then, the signal received  at the $k$-th user  can be expressed as
\begin{align}
    y_k=\sum\limits_{m=1}^{K} \sum\limits_{i=1}^{K}  {\bf g}_{\text{I2U},mk}^T { \bm \Theta_{m}} {\bf G}_{\text{B2I},m} {\bf w}_i s_i+n_k,
\end{align}
where $n_{k} \sim \mathcal{CN} \left(0,\sigma_0^2\right)$ is the noise at the $k$-th user.  The phase shift matrix of the $m$-th IRS is given by  ${\bm \Theta}_m= {\text {diag}} \left( {\bm \xi}_m\right) {  \in \mathbb{C}^{M \times M} } $
with the phase shift  beam $ {\bm \xi}_m={[ e^{j\vartheta_{m,1}},..., e^{j\vartheta_{m,n}},...,  e^{j\vartheta_{m,M}}]}^T {  \in \mathbb{C}^{M \times 1} }$.

\section{Location based angle information acquisition}\label{s2}
To facilitate the design of the transmit beam ${\bf w}_i$ and phase shift beam ${\bm \xi}_m$, CSI of individual channels is necessary. In this section, we exploit the user location information provided by the GPS to obtain angle information of the LOS path. Furthermore, due to the fact that IRS locations are fixed, we assume that they are perfectly known by the BS.

Without loss of generality, we assume the  BS is located  at the origin $(0,0,0)$, and the $m$-th IRS is located at $(x_{\text{I},m}, y_{\text{I},m}, z_{\text{I},m})$, then the effective AOD from the BS to the $m$-th IRS can be computed as
\begin{align}
\vartheta_{\text{y-B2I},m} &=-\frac{y_{\text{I},m}}{ d_{\text{B2I},m}},
\end{align}
where $ d_{\text{B2I},m}=\sqrt{x_{\text{I},m}^2+y_{\text{I},m}^2+z_{\text{I},m}^2}$ is the distance between the BS and the $m$-th IRS.

Similarly, the effective AOA at the $m$-th IRS can be  calculated as
\begin{align}
\vartheta_{\text{y-B2Ia},m} &=\frac{y_{\text{I},m}}{ d_{\text{B2I},m}}.
\end{align}

Due to user mobility or other unfavorable conditions, the user location information  provided by the GPS is imperfect. Specifically, the accurate location of user $k$, i.e., $( x_{\text{U},k}, y_{\text{U},k}, z_{\text{U},k})$, is uniformly distributed within a sphere with the radius $\Upsilon$ and center point  $(\hat x_{\text{U},k},\hat  y_{\text{U},k},\hat z_{\text{U},k})$, where $(\hat x_{\text{U},k},\hat  y_{\text{U},k},\hat z_{\text{U},k})$ is the estimated location of user $k$, provided by the GPS.

The $m$-th IRS calculates its effective AOD from itself to the $k$-th  user as
\begin{align}\label{E1}
\hat\vartheta_{\text{y-I2U},mk} &=\frac{y_{\text{I},m}-\hat y_{\text{U},k}}{\hat d_{\text{I2U},mk}},
\end{align}
where $\hat d_{\text{I2U},mk}=\sqrt{(x_{\text{I},m}-\hat x_{\text{U},k})^2+(y_{\text{I},m}-\hat y_{\text{U},k})^2+(z_{\text{I},m}-\hat z_{\text{U},k})^2}$ is the  distance between the $m$-th IRS and the $k$-th user.

\begin{proposition} \label{p1}
The   effective AOD from the $m$-th IRS  to the $k$-th  user can be decomposed as
\begin{align}
   \vartheta_{\text{y-I2U},mk} =
   \hat\vartheta_{\text{y-I2U},mk} +\epsilon_{\text{y-I2U},mk},\label{E5}
\end{align}
where
\begin{align}
 & \epsilon_{\text{y-I2U},mk}
  =  \frac{ \left( \hat\vartheta^2_{\text{y-I2U},mk}-1\right) \Delta y_{\text{U},k} + \hat\vartheta_{\text{y-I2U},mk} \hat\vartheta_{\text{z-I2U},mk} \Delta z_{\text{U},k}+\hat\vartheta_{\text{y-I2U},mk}\hat\vartheta_{\text{z-I2U},mk} \Delta z_{\text{U},k} }{\hat d_{\text{I2U},mk}},
\end{align}
where   $\hat \vartheta_{\text{x-I2U},mk}  \triangleq \frac{x_{\text{I},m}-\hat x_{\text{U},k}}{\hat d_{\text{I2U},mk}}$, $\hat \vartheta_{\text{z-I2U},mk}  \triangleq \frac{z_{\text{I},m}-\hat z_{\text{U},k}}{\hat d_{\text{I2U},mk}}$,
 $\Delta x_{\text{U},k}=x_{\text{U},k}-\hat x_{\text{U},k}$,
  $\Delta y_{\text{U},k}=y_{\text{U},k}-\hat y_{\text{U},k}$ and
 $\Delta z_{\text{U},k}=z_{\text{U},k}-\hat z_{\text{U},k}$ are location errors along $x$, $y$  and $z$ axes, respectively.
\end{proposition}
\begin{proof}
See Appendix \ref{Ax1}.
\end{proof}

Furthermore, the estimation error $\epsilon_{\text{y-I2U},mk}$ has the following distribution.
 \begin{theorem}\label{t1}
 The PDF of estimation error  $\epsilon_{\text{y-I2U},mk}$ is given by
 \begin{align}
   f_{\epsilon_{\text{y-I2U},mk}} \left( x\right) = - \frac{3 {\hat{d}}^3_{\text{I2U},mk} } {4 \Upsilon^3} {\Phi}^{-3}_{\text{y-I2U},mk} x^2  +\frac{3 {\hat{d}}_{\text{I2U},mk}}{4 \Upsilon} {\Phi}^{-1}_{\text{y-I2U},mk}, \  |x|\le \frac{\Upsilon}{\hat{d}_{\text{I2U},mk}} {\Phi}_{\text{y-I2U},mk},
 \end{align}
 where
 \begin{align}
 \Phi_{\text{y-I2U},mk} \triangleq \sqrt {{\left( \hat\vartheta^2_{\text{y-I2U},mk}-1 \right)}^2+\hat\vartheta^2_{\text{y-I2U},mk} \hat\vartheta^2_{\text{z-I2U},mk}+\hat\vartheta^2_{\text{y-I2U},mk} \hat\vartheta^2_{\text{x-I2U},mk}}.
 \end{align}

 The  mean and  variance of  $\epsilon_{\text{y-I2U},mk}$ are given by
 \begin{align}
    &\mu_{\text{y-I2U},mk}=\mathbb{E}\left\{ \epsilon_{\text{y-I2U},mk}\right\}=0,\\
  & \sigma_{\text{y-I2U},mk}^2=\mathbb{E}\left\{\epsilon^2_{\text{y-I2U},mk}\right\}= \frac{\Upsilon^2}{5\hat{d}^2_{\text{I2U},mk}} \Phi^2_{\text{y-I2U},mk}.
 \end{align}
 \end{theorem}

 \begin{proof}
See Appendix \ref{At1}.
 \end{proof}

 \begin{remark}
Theorem \ref{t1} indicates that  angle estimation accuracy is greatly affected by the user location uncertainty as well as the distance between the IRS and the user. Specifically,  the variance of the angle estimation error
is proportional to  the ratio $\frac{\Upsilon^2}{\hat{d}^2_{\text{I2U},mk}}$, implying that increasing the distance between the IRS and the user can compensate for the adverse effect of user location uncertainty.
 \end{remark}

 \section{Design of transmit and phase shift beams} \label{s3}
After obtaining the angle information, the BS uses the estimated angles to design transmit and phase shift beams. Since transmit and phase shift beams are coupled, the resultant optimization problem is non-convex, hence the global optimal solution is in general intractable. As such, many sub-optimal algorithms have been proposed such as alternating optimization. However, due to the large number of reflecting elements, the computational complexity of these algorithms is very high. Motivated by this, instead of aiming for the global optimal solution, we perform low-complexity local optimization at the BS and each IRS, separately, where closed-form solutions can be obtained.

Specifically, the BS utilizes the angle information of the BS-IRS link to design the transmit beam. Without loss of generality, we assume the $k$-th user is assisted by the $k$-th IRS. Thus, the transmit beam for the $k$-th user should be aligned to the $k$-th IRS. As such, the transmit beam is designed as
\begin{align}
{\bf w}_{k}=\sqrt{\frac{\eta_k \rho_d}{N}} {\bf{a}}_{\text{B2I},k}^*,
\end{align}
where ${\bf{a}}_{\text{B2I},k}\triangleq {\bf a} \left( \vartheta_{\text{y-B2I},k} \right){  \in \mathbb{C}^{N \times 1} }$, $\rho_d$ is the transmit power of the BS and $0 <\eta_k<1$ is the power control coefficient.

The estimated angle information of the IRS-user link is used to design the phase shift beam. For the $k$-th user, we aim to maximize its received signal via the $k$-th IRS by optimizing the phase shift beam $\bm {\xi}_k$.
Specifically, the optimization problem is formulated as
\begin{align} \label{Ebeam2}
  \begin{array}{ll}
    \max\limits_{{\bm \Theta}_k}
& \left| {\bf g}_{\text{I2U},kk}^T { \bm \Theta_{k}} {\bf G}_{\text{B2I},k} {\bf w}_k \right|^2,  \\
 \operatorname{s.t.} & \begin{array}[t]{lll}
             \left|{\left[ {\bf \Theta}_k \right]}_{ii}\right| =1,i=1,...,M.
           \end{array}
  \end{array}
\end{align}

According to the rule that ${\bf E}^T {\bf X} {\bf F} = {\bf x}^T \left( {\bf E} \odot {\bf F} \right)$ with ${\bf X}={\text {diag}} \left({\bf x} \right)$, we have
 \begin{align} \label{Ebeam}
  {\bf g}_{\text{I2U},kk}^T { \bm \Theta_{k}} {\bf G}_{\text{B2I},k} {\bf w}_k={ \bm \xi_{k}^T} \left({\bf g}_{\text{I2U},kk} \odot {\bf G}_{\text{B2I},k} {\bf w}_k \right)={ \bm \xi_{k}^T} \left({\bf g}_{\text{I2U},kk} \odot {\bf b}_{\text{B2I},k} \right) {\bf a}^T_{\text{B2I},k} {\bf w}_k,
  \end{align}
  where ${\bf b}_{\text{B2I},k}\triangleq {\bf b} \left( \vartheta_{\text{y-B2Ia},k} \right) {  \in \mathbb{C}^{M \times 1} } $.

 Based on the above equation, the objective function in (\ref{Ebeam2}) becomes
\begin{align}
  \left|{ \bm \xi_{k}^T} \left({\bf g}_{\text{I2U},kk} \odot {\bf b}_{\text{B2I},k} \right) \right|^2 \left|{\bf a}^T_{\text{B2I},k} {\bf w}_k \right|^2.
\end{align}

Noticing that $\left|{\bf a}^T_{\text{B2I},k} {\bf w}_k \right|^2$ is a constant independent of ${\bm \xi}_k$ , the optimization problem is equivalent to
 \begin{align}
  \begin{array}{ll}
    \max\limits_{{\bm \xi}_k^T}
&\left| { \bm \xi_{k}^T } \left({\bf g}_{\text{I2U},kk} \odot {\bf b}_{\text{B2I},k} \right)\right|^2,  \\
 \operatorname{s.t.} & \begin{array}[t]{lll}
             \left| \left[ {\bm \xi}_{k} \right]_{i}\right| =1,i=1,...,M.
           \end{array}
  \end{array}
\end{align}

Obviously, the solution for the above optimization problem is
\begin{align}
 {\bm \xi}_k=\left({\bf g}_{\text{I2U},kk} \odot {\bf b}_{\text{B2I},k} \right)^*.
\end{align}

Using the estimated angles, we design the phase shift beam  as
\begin{align}
    {\bm \xi}_k=\left({\hat{\bar{\bf g}}}_{\text{I2U},kk} \odot {\bf b}_{\text{B2I},k} \right)^*,
\end{align}
where ${\hat{\bar{\bf g}}}_{\text{I2U},mk} = {\bf b} \left(\hat\vartheta_{\text{y-I2U},mk}\right) { \in \mathbb{C}^{M \times 1} } $.
%

\section{Achievable Rate Analysis} \label{s4}
In this section, we present a detailed investigation on the achievable rate of the multi-IRS system.  Without loss of generality, let us focus on the achievable rate of the $k$-th user. We consider the realistic case where the $k$-th user does not have access to the instantaneous CSI of the effective channel gain.  As such, we can rewrite ${{y}}_{k}$ as
 \begin{align}
      y_k=\underbrace{ \mathbb{E}\left\{  {\bf g}_{k}^T   {\bf w}_k \right\} s_k}_{\text{desired signal}}
+\underbrace{{\left({\bf g}_{k}^T   {\bf w}_k - \mathbb{E}\left\{  {\bf g}_{k}^T   {\bf w}_k  \right\} \right)}s_k}_{\text{leakage signal}}
 +\underbrace{{\bf g}_{k}^T \sum\limits_{i \ne k}^{K} {\bf w}_i s_i}_{\text{inter-user interference}}
   +\underbrace{n_k}_{\text{noise}},
 \end{align}
 where we define the equivalent channel of the $k$-the user as
\begin{align}
     {\bf g}_k^T= \sum\limits_{m=1}^{K}
    {\bf g}_{\text{I2U},mk}^T { \bm \Theta_{m}} {\bf G}_{\text{B2I},m}.
\end{align}

 Invoking the result in \cite{marzetta2016fundamentals}, the achievable rate  of the $k$-th user is given by
 \begin{align}
 R_k=\log_2 \left( 1+\frac{A_k}{B_k+\sum\limits_{i\ne k}^{K} C_{k,i}+\sigma_0^2 } \right),
 \end{align}
 where
 \begin{align}
    & A_k \triangleq { \left|\mathbb{E}\left\{  {\bf g}_{k}^T   {\bf w}_k \right\}\right|}^2,\\
    & B_k \triangleq \mathbb{E} \left\{ {\left| {{{\bf g}_{k}^T   {\bf w}_k - \mathbb{E}\left\{  {\bf g}_{k}^T   {\bf w}_k  \right\} }} \right|}^2 \right\},\\
    &C_{k,i}\triangleq  \mathbb{E}\left\{  {\left| {{\bf g}_{k}^T {\bf w}_i } \right|}^2\right\},
 \end{align}
 denote the desired signal power, leakage power and interference, respectively.

 To derive the expression of $R_k$, we first give the following Lemma:
 \begin{lemma} \label{L2}
 The correlation coefficient between $\varrho_{\text{I2U},mk,s}  \triangleq e^{j \pi
 \left(s-1\right)\epsilon_{\text{y-I2U},mk} }$ and $ \varrho_{\text{I2U},nk,l}, n\ne m$,  is given by
\begin{align}
& \zeta_{\text{y-I2U},mk,nk,sl}
 \triangleq
  \mathbb{E}\left\{
   \varrho_{\text{I2U},mk,s}
  \varrho_{\text{I2U},nk,l}^* \right\},\\
&=
\begin{cases}
1, & s=1\ \text{and} \  l=1\\
 \frac{3}{\varpi_{ k,mn,sl}^2} \left(
\frac{\sin \varpi_{ k,mn,sl} }{\varpi_{ k,mn,sl}} -\cos\varpi_{ k,mn,sl} \right), \nonumber & \text{else}
\end{cases}
  \end{align}
 where
\begin{align}
 &\varpi_{ k,mn,sl} \triangleq \pi \Upsilon \sqrt { a^2_{k,mn,sl} +b^2_{k,mn,sl}+c^2_{k,mn,sl} }, \\
  & a_{k,mn,sl} \triangleq
   \left(s-1\right) \frac{ \hat\vartheta^2_{\text{y-I2U},mk}-1}{\hat d_{\text{I2U},mk}} - \left(l-1\right) \frac{ \hat\vartheta^2_{\text{y-I2U},nk}-1}{\hat d_{\text{I2U},nk}}, \\
  & b_{k,mn,sl} \triangleq
    \left(s-1\right)\frac{\hat\vartheta_{\text{y-I2U},mk} \hat\vartheta_{\text{z-I2U},mk}}{\hat d_{\text{I2U},mk}}- \left(l-1\right)\frac{\hat\vartheta_{\text{y-I2U},nk} \hat\vartheta_{\text{z-I2U},nk}}{\hat d_{\text{I2U},nk}},\\
    &c_{k,mn,sl}\triangleq
    \left(s-1\right)  \frac{\hat\vartheta_{\text{y-I2U},mk}\hat\vartheta_{\text{x-I2U},mk} }{\hat d_{\text{I2U},mk}}- \left(l-1\right)  \frac{\hat\vartheta_{\text{y-I2U},nk}\hat\vartheta_{\text{x-I2U},nk} }{\hat d_{\text{I2U},nk}}.
\end{align}
 \end{lemma}

 \begin{proof}
 See Appendix \ref{AL2}.
 \end{proof}

 \begin{theorem} \label{t2}
The achievable rate of the $k$-th user is given by
\begin{align}
 R_k=\log_2 \left( 1+\frac{A_k}{B_k+\sum\limits_{i\ne k}^{K} C_{k,i}+\sigma_0^2 } \right),
 \end{align}
 where
 \begin{align}
    & A_k=\frac{\eta_k \rho_d}{N } \left| \sum\limits_{m= 1}^{K}  \sum\limits_{s=1}^{M} \sqrt{\beta_{\text{B2I2U},mk}}
 {\bf a}^T_{\text{B2I},m} {\bf a}_{\text{B2I},k}^*
   \zeta_{\text{y-I2U},mk,s}
{\left[ {\hat{\bar{\bf g}}}_{\text{I2U},mk} \right]}_s  {\left[{\hat{\bar{\bf g}}}_{\text{I2U},mm} \right]}_s^*
   \right|^2,\\
    &  C_{k,i}=\frac{\eta_i \rho_d}{N} \sum\limits_{m=1}^{K} \left(M \frac{\beta_{\text{B2I2U},mk} }{v_{\text{I2U},mk}} \left|{\bf a}^T_{\text{B2I},m}{\bf a}_{\text{B2I},i}^* \right|^2+
   MN\frac{\beta_{\text{B2I2U},mk}}{v_{\text{B2I},m}v_{\text{I2U},mk} }
   +MN \frac{\beta_{\text{B2I2U},mk}}{v_{\text{B2I},m}}\right)
   \\
  &+\frac{\eta_i \rho_d}{N}\sum\limits_{m=1}^{K}
  \sum\limits_{s=1}^{M}   \sum\limits_{l=1}^{M} \beta_{\text{B2I2U},mk}
  \left| {\bf a}^T_{\text{B2I},m} {\bf a}_{\text{B2I},i}^* \right|^2
 \zeta_{\text{y-I2U},mk,sl}
{\left[ {\hat{\bar{\bf g}}}_{\text{I2U},mk} \right]}_s  {\left[{\hat{\bar{\bf g}}}_{\text{I2U},mm} \right]}_s^*
{\left[ {\hat{\bar{\bf g}}}_{\text{I2U},mk} \right]}_l^*  {\left[{\hat{\bar{\bf g}}}_{\text{I2U},mm} \right]}_l
\nonumber \\
&+\frac{\eta_i \rho_d}{N}\sum\limits_{m=1}^{K} \sum\limits_{n\ne m}^{K}
   \sum\limits_{s=1}^{M}   \sum\limits_{l=1}^{M}
   {\bf a}^T_{\text{B2I},m} {\bf a}_{\text{B2I},i}^* \sqrt{\beta_{\text{B2I2U},mk}}  {\bf a}^H_{\text{B2I},n}  {\bf a}_{\text{B2I},i} \sqrt{\beta_{\text{B2I2U},nk}}
\nonumber \\
&\times
\zeta_{\text{y-I2U},mk,nk,sl}
{\left[ {\hat{\bar{\bf g}}}_{\text{I2U},mk} \right]}_s  {\left[{\hat{\bar{\bf g}}}_{\text{I2U},mm} \right]}_s^*
{\left[ {\hat{\bar{\bf g}}}_{\text{I2U},nk} \right]}_l^*  {\left[{\hat{\bar{\bf g}}}_{\text{I2U},nn} \right]}_l \nonumber,  \\
 & B_k =C_{k,k}-A_k,
 \end{align}
 where  $\beta_{\text{B2I2U},mk} \triangleq \frac{\alpha_{\text{B2I},m} v_{\text{B2I},m} \alpha_{\text{I2U},mk} v_{\text{I2U},mk} }{  (v_{\text{B2I},m}+1 ) (v_{\text{I2U},mk}+1 )}$ and
\begin{align}
\zeta_{\text{y-I2U},mk,s}\! =\!
&\begin{cases}
 \frac{3}{\varpi_{ \text{y-I2U},mk,s}^2} \left(
\frac{\sin \varpi_{  \text{y-I2U},mk,s} }{\varpi_{  \text{y-I2U},mk,s}} -\cos\varpi_{  \text{y-I2U},mk,s} \right) & s \!\ne\! 1\\
1& s\! =\!1
\end{cases},\\
\zeta_{\text{y-I2U},mk,sl}\! =\!
&\begin{cases}
\frac{3}{\varpi_{ \text{y-I2U},mk,sl}^2} \left(
\frac{\sin \varpi_{  \text{y-I2U},mk,sl} }{\varpi_{  \text{y-I2U},mk,sl}} -\cos\varpi_{  \text{y-I2U},mk,sl} \right) & s \!\ne\! l\\
1& s\! =\!l
\end{cases},
\end{align}
with
$
   \varpi_{\text{y-I2U},mk,s} \triangleq \frac{\pi \left(s-1\right) {\Phi}_{\text{y-I2U},mk} \Upsilon}{ {\hat{d}}_{\text{I2U},mk}}
$
and $
    \varpi_{\text{y-I2U},mk,sl} \triangleq \frac{\pi \left(s-l\right){\Phi}_{\text{y-I2U},mk} \Upsilon}{ {\hat{d}}_{\text{I2U},mk}}
$.
\end{theorem}

\begin{proof}
See Appendix  \ref{At2}.
\end{proof}

Theorem \ref{t2} presents a  closed-form expression for the achievable rate,  which quantifies the impacts of some key parameters, such as antenna number, user number, the number of reflecting elements,  Rician K-factor and user location uncertainty.
For instance, $R_k$ is an increasing function with respect to the Rician K-factor, because a small Rician K-factor implies more interference caused by NLOS paths. Also, increasing the number of users would degrade the individual rate, due to stronger inter-user interference.
In addition, as $\Upsilon$ becomes larger, $\zeta_{\text{y-I2U},mk,s} \to 0$ and thus $R_{k}$ gradually goes to zero, indicating that the achievable rate would  significantly degrade as user location uncertainty becomes larger.
Moreover, we can see that $C_{k,i}$ increases with the correlation coefficient between ${\bf a}_{\text{B2I},i}$ and ${\bf a}_{\text{B2I},m}, m\ne i$,
implying that IRSs should be  deployed
at distinct directions relative to the BS
to reduce the interference from multiple IRSs. Ideally, it is desired that ${\bf a}_{\text{B2I},m}^T {\bf a}_{\text{B2I},i}^* \to 0, m \ne i$.

We now consider some special cases, where simplified expressions are available.
\subsection{Ideal directions (relative to the BS) of IRSs }
\begin{proposition} \label{p2}
To reduce the interference from multiple IRSs, ideally
the directions  (relative to the BS) of any two IRSs should  satisfy
$\left| \vartheta_{\text{B2I},m} -\vartheta_{\text{B2I},i} \right|= \frac{2n}{N}, i\ne m, n \in \{1,...,N-1\}$.
\end{proposition}

\begin{proof}
Let ${\bf a}_{\text{B2I},m}^T {\bf a}_{\text{B2I},i}^* = 0, m \ne i$. Noticing that
\begin{align}
{\bf a}_{\text{B2I},m}^T {\bf a}_{\text{B2I},i} = \frac{\sin \frac{N \pi (\vartheta_{\text{B2I},m} -\vartheta_{\text{B2I},i})}{2} }{\sin \frac{ \pi (\vartheta_{\text{B2I},m} -\vartheta_{\text{B2I},i})}{2} }
e^{j \pi (N-1) \frac{\vartheta_{\text{B2I},m} -\vartheta_{\text{B2I},i}}{2}},
\end{align}
we can obtain the desired result.
\end{proof}

 Proposition \ref{p2} indicates that it is possible to eliminate the interference from other IRSs by deploying them at proper directions (relative to the BS).

\begin{corollary} \label{c1}
With ideal directions (relative to the BS) of IRSs, i.e., ${\bf a}_{\text{B2I},m}^T {\bf a}_{\text{B2I},i} \to 0, m \ne i$, the achievable rate is given by
\begin{align}
 R_k=\log_2 \left( 1+\frac{A_k}{B_k+\sum\limits_{i\ne k}^{K} C_{k,i}+\sigma_0^2 } \right),
 \end{align}
 where
 \begin{align}
    & A_k=N \eta_k \rho_d \beta_{\text{B2I2U},kk}  \left|   \sum\limits_{s=1}^{M}
   \zeta_{\text{y-I2U},kk,s}
   \right|^2,\\
    &  C_{k,i}=
   N  M \eta_i \rho_d \frac{\beta_{\text{B2I2U},ik} }{v_{\text{I2U},ik}}+
     M \eta_i \rho_d \sum\limits_{m=1}^{K} \left(
  \frac{\beta_{\text{B2I2U},mk}}{v_{\text{B2I},m}v_{\text{I2U},mk} }
   + \frac{\beta_{\text{B2I2U},mk}}{v_{\text{B2I},m}}\right)
   \\
  &+N \eta_i \rho_d
  \sum\limits_{s=1}^{M}   \sum\limits_{l=1}^{M} \beta_{\text{B2I2U},ik}
 \zeta_{\text{y-I2U},ik,sl}
{\left[ {\hat{\bar{\bf g}}}_{\text{I2U},ik} \right]}_s  {\left[{\hat{\bar{\bf g}}}_{\text{I2U},ii} \right]}_s^*
{\left[ {\hat{\bar{\bf g}}}_{\text{I2U},ik} \right]}_l^*  {\left[{\hat{\bar{\bf g}}}_{\text{I2U},ii} \right]}_l,
\nonumber \\
 & B_k = N  M \eta_k \rho_d \frac{\beta_{\text{B2I2U},kk} }{v_{\text{I2U},kk}}+
     M \eta_k \rho_d \sum\limits_{m=1}^{K} \left(
  \frac{\beta_{\text{B2I2U},mk}}{v_{\text{B2I},m}v_{\text{I2U},mk} }
   + \frac{\beta_{\text{B2I2U},mk}}{v_{\text{B2I},m}}\right)
   \\
  &+N \eta_k \rho_d \beta_{\text{B2I2U},kk}
  \sum\limits_{s=1}^{M}   \sum\limits_{l=1}^{M}
 \left(  \zeta_{\text{y-I2U},kk,sl} - \zeta_{\text{y-I2U},kk,s} \zeta_{\text{y-I2U},kk,l}^* \right). \nonumber
 \end{align}

\end{corollary}

\begin{proof}
Starting from Theorem \ref{t2}, we can obtain the desired result.
\end{proof}

From Corollary \ref{c1}, we can observe that the desire signal $A_k$ is proportional to $\beta_{\text{B2I2U},kk}$, indicating that a user should be associated with the nearest IRS. Besides, it can be seen that $C_{k,i}$ increases proportionally with $\beta_{\text{B2I2U},ki}, i\ne k$, which implies that the interference from other IRSs (with indexes $i\ne k$) would be reduced if they are far away from the $k$-th user.
According to the above two observations, we conclude that  IRSs should be deployed far away from each other.
 Moreover, we can see that the desired signal power is proportional to the number of BS antennas, indicating the benefit of applying multiple antennas.

\subsection{Perfect information of user locations}
\begin{corollary} \label{c2}
With perfect user location information, i.e., $\Upsilon \to 0$, the achievable rate is given by
\begin{align}
 R_k=\log_2 \left( 1+\frac{A_k}{B_k+\sum\limits_{i\ne k}^{K} C_{k,i}+\sigma_0^2 } \right),
 \end{align}
 where
 \begin{align}
    & A_k=N M^2 \eta_k \rho_d \beta_{\text{B2I2U},kk}  ,\\
    &  C_{k,i}=
   N  M \eta_i \rho_d \frac{\beta_{\text{B2I2U},ik} }{v_{\text{I2U},ik}}+
     M \eta_i \rho_d \sum\limits_{m=1}^{K} \left(
  \frac{\beta_{\text{B2I2U},mk}}{v_{\text{B2I},m}v_{\text{I2U},mk} }
   + \frac{\beta_{\text{B2I2U},mk}}{v_{\text{B2I},m}}\right)
   \\
  &+N \eta_i \rho_d \beta_{\text{B2I2U},ik}
 \left| \frac{\sin \frac{M \pi (\hat\vartheta_{\text{I2U},ik} -\hat\vartheta_{\text{I2U},ii})}{2} }{\sin \frac{ \pi (\hat\vartheta_{\text{I2U},ik} -\hat\vartheta_{\text{I2U},ii})}{2} } \right|^2,
\nonumber \\
 & B_k = N  M \eta_k \rho_d \frac{\beta_{\text{B2I2U},kk} }{v_{\text{I2U},kk}}+
     M \eta_k \rho_d \sum\limits_{m=1}^{K} \left(
  \frac{\beta_{\text{B2I2U},mk}}{v_{\text{B2I},m}v_{\text{I2U},mk} }
   + \frac{\beta_{\text{B2I2U},mk}}{v_{\text{B2I},m}}\right).
 \end{align}

\end{corollary}

\begin{proof}
We first consider the following limit:
\begin{align} \label{E14}
    \lim_{\varpi \to 0} \zeta= \lim_{\varpi \to 0} \frac{3}{\varpi^2} \left(
\frac{\sin \varpi }{\varpi} -\cos\varpi \right).
\end{align}

 Using the Taylor expansions of $\sin \varpi $ and
 $\cos \varpi $ at $\varpi=0$, we have
 \begin{align}
     \lim_{\varpi \to 0} \sin \varpi= \varpi-\frac{1}{6}\varpi^3,\\
     \lim_{\varpi \to 0} \cos \varpi=1-\frac{1}{2}\varpi^2,
 \end{align}
 based on which, (\ref{E14}) can be expressed as
 \begin{align} \label{limit}
    \lim_{\varpi \to 0} \zeta= \lim_{\varpi \to 0} \frac{3}{\varpi^2} \left\{
\frac{\varpi-\frac{1}{6}\varpi^3}{\varpi} - \left(1-\frac{1}{2}\varpi^2\right)  \right\}=1.
\end{align}

As $\Upsilon \to \infty$, $\varpi_{\text{y-I2U},mk,s} \to 0$, $\varpi_{\text{y-I2U},mk,sl} \to 0$ and $\varpi_{k,mn,sl} \to 0$.

Invoking (\ref{limit}), we have
$\zeta_{\text{y-I2U},mk,s} \to 1$, $\zeta_{\text{y-I2U},mk,sl} \to 1$ and $\zeta_{\text{y-I2U},mk,nk,sl} \to 1$.

As such, simplifying the achievable rate given by Corollary \ref{c1} yields the desired result.
\end{proof}

Corollary \ref{c2} shows that as  $\Upsilon$ goes to zero, the achievable rate  converges to a limit, due to the vanished angle estimation error. Moreover, we can see that the desired signal power becomes proportional to $N M^2$.
The $N$-fold gain is achieved by the transmit beamforming, while the $M^2$-fold gain comes from the fact that the IRS not only achieves the phase shift beamforming gain in the IRS-user link, but also captures an inherent aperture gain by collecting more signal power in the BS-IRS link.

To gain more insights, we now look into some asymptotic regime.
\subsubsection{A large number of reflecting elements}
\begin{corollary} \label{c3}
With a large number of reflecting elements, i.e., $M \to \infty$, the achievable rate is given by
\begin{align}
 R_k=\log_2 \left( 1+\frac{N M \eta_k  \beta_{\text{B2I2U},kk}}{ N\sum\limits_{i= 1}^{K}    \eta_i  \frac{\beta_{\text{B2I2U},ik} }{v_{\text{I2U},ik}}
 +  \sum\limits_{m=1}^{K} \left(
  \frac{\beta_{\text{B2I2U},mk}}{v_{\text{B2I},m}v_{\text{I2U},mk} }
   + \frac{\beta_{\text{B2I2U},mk}}{v_{\text{B2I},m}}\right)} \right).
 \end{align}

\end{corollary}

\begin{proof}
Starting from Corollary \ref{c2} and neglecting the insignificant terms, the desired result can be obtained.
\end{proof}

From Corollary \ref{c3}, we can see that the achievable rate is mainly determined by the number of reflecting elements, antenna number, Rician K-factor and power allocation coefficients.
For instance, the SINR is proportional to the number of reflecting elements, indicating the benefit of deploying a large number of reflecting elements. However, in practice, to control the cost of IRS deployment, the number of reflecting elements can not be infinite. Since the achievable rate grows logarithmically with the number of reflecting elements, the benefit of increasing the number of reflecting elements becomes less significant in the regime with large number of reflecting elements. Thus, caution should be exercised for choosing the number of IRS elements in order to achieve a fine balance between the deployment cost and the achievable rate.

Also, the achievable rate is an increasing function with respect to the number of BS antennas.

\subsubsection{A large number of antennas}
\begin{corollary} \label{c4}
With a large number of antennas, i.e., $N \to \infty$, the achievable rate is given by
\begin{align}
 R_k=\log_2 \left( 1+\frac{M^2 \eta_k  \beta_{\text{B2I2U},kk} }
 { M\sum\limits_{i=1}^{K}  \eta_i  \frac{\beta_{\text{B2I2U},ik} }{v_{\text{I2U},ik}}
 +\sum\limits_{i\ne k}^{K}
 \eta_i  \beta_{\text{B2I2U},ik}
 \left| \frac{\sin \frac{M \pi (\hat\vartheta_{\text{I2U},ik} -\hat\vartheta_{\text{I2U},ii})}{2} }{\sin \frac{ \pi (\hat\vartheta_{\text{I2U},ik} -\hat\vartheta_{\text{I2U},ii})}{2} } \right|^2} \right).
 \end{align}
\end{corollary}

\begin{proof}
Starting from Corollary \ref{c2}  and ignoring the terms that do not scale with $N$,  we can obtain the desired result.
\end{proof}

Corollary \ref{c4} shows  that as the number of antennas becomes larger, the achievable rate becomes independent of the antenna number and gradually converges to a limit, which is  mainly determined by the number of reflecting elements, and Rician K-factor of IRS-user channels.
Specifically, the achievable rate grows with the number of reflecting elements, and increases as the Rician K-factor of IRS-user channels becomes larger.
Moreover, the interference caused by NLOS paths between the BS and IRSs vanishes, indicating that increasing  the antenna number can compensate for the adverse effect of NLOS paths between the BS and IRSs.

\subsubsection{No NLOS paths}
\begin{corollary} \label{c5}
Without NLOS paths, i.e., $v_{\text{I2U},mk} \to \infty$ and $v_{\text{B2I},m} \to \infty$, the achievable rate is given by
\begin{align}
 R_k=\log_2 \left( 1+\frac{M^2 N \eta_k  \beta_{\text{B2I2U},kk} }
 {N \sum\limits_{i\ne k}^{K}
 \eta_i  \beta_{\text{B2I2U},ik}
 \left| \frac{\sin \frac{M \pi (\hat\vartheta_{\text{I2U},ik} -\hat\vartheta_{\text{I2U},ii})}{2} }{\sin \frac{ \pi (\hat\vartheta_{\text{I2U},ik} -\hat\vartheta_{\text{I2U},ii})}{2} } \right|^2 +\sigma_0^2} \right).
 \end{align}
\end{corollary}

\begin{proof}
Starting from Corollary \ref{c2}  and ignoring the terms that go to zero as $v_{\text{I2U},mk} \to \infty$ and $v_{\text{B2I},m} \to \infty$,  the desired result can be obtained.
\end{proof}

Corollary \ref{c5} shows that without NLOS paths, the SINR is mainly determined by the number of reflecting elements and the antenna number. In particular, the SINR becomes nearly proportional to $M^2$.

\subsubsection{The impact of user directions relative to a IRS}
\begin{proposition} \label{User_Direction}
The $i$-th ($i \ne k$) IRS would cause severe  interference (nearly proportional to $M^2$) to the $k$-th user, if the directions (relative to the $i$-th IRS ) of  the $k$-th user and the $i$-th user are similar, i.e., $\hat\vartheta_{\text{I2U},ik} -\hat\vartheta_{\text{I2U},ii} \to 0$.
\end{proposition}
\begin{proof}
  As $\hat\vartheta_{\text{I2U},ik} -\hat\vartheta_{\text{I2U},ii} \to 0$, the second term in $C_{k,i}$ can be approximated as
  $N M^2 \eta_i \rho_d \beta_{\text{B2I2U},ik}$, which increases proportionally with $M^2$.
\end{proof}

Proposition \ref{User_Direction} indicates that a user would suffer more interference from a IRS, if this user is in the similar direction (relative to the IRS) to the user assisted by this IRS.


\section{Power control}\label{PC}
 To improve the energy efficiency of the system, and facilitates interference coordination, we propose a low-complexity power control algorithm to minimize the transmit power, subject to individual user rate requirements.

Specifically, the optimization problem can be formulated as
\begin{align} {\label{E13}}
  \begin{array}{ll}
    \min\limits_{\left\{ {\bf p} \right\}}
& \rho_d,  \\
 \text{s.t.} & \begin{array}[t]{lll}
              \sum\limits_{k=1}^{K} p_k=\rho_d,\\
             R_k \ge \bar R_k,k=1,...,K,\\
             p_k \ge 0, k=1,...,K,
           \end{array}
  \end{array}
\end{align}
where ${\bf p}\triangleq [p_1,..., p_k...,p_K]^T$ with $p_k\triangleq \eta_k \rho_d$, and $\bar R_k$ is the required minimum rate of the $k$-th user.

Substituting the equality constraint into the objective function, we rewrite the above optimization problem as
\begin{align} {\label{E13}}
  \begin{array}{ll}
    \min\limits_{\left\{ {\bf p} \right\}}
&  \sum\limits_{k=1}^{K} p_k,  \\
 \text{s.t.} & \begin{array}[t]{lll}
             R_k \ge \bar R_k,k=1,...,K,\\
             p_k \ge 0, k=1,...,K.
           \end{array}
  \end{array}
\end{align}

According to the derived achievable rate in Theorem \ref{t2}, the above optimization problem can be expressed as
\begin{align} {\label{E13}}
  \begin{array}{ll}
    \min\limits_{\left\{ {\bf p} \right\}}
&  \sum\limits_{k=1}^{K} p_k,  \\
 \text{s.t.} & \begin{array}[t]{lll}
             \sum\limits_{i \ne k}^{K}
             \left(2^{\bar R_k}-1 \right) \bar c_{k,i} p_i+\left\{ \left(2^{\bar R_k}-1 \right) \bar b_k-\bar a_k\right\} p_k+\left(2^{\bar R_k}-1 \right)  \sigma_0^2 \le 0,
             k=1,...,K,\\
             p_k \ge 0, k=1,...,K,
           \end{array}
  \end{array}
\end{align}
where
\begin{align}
   & \bar a_k \triangleq \frac{1}{N } \left| \sum\limits_{m= 1}^{K}  \sum\limits_{s=1}^{M} \sqrt{\beta_{\text{B2I2U},mk}}
 {\bf a}^T_{\text{B2I},m} {\bf a}_{\text{B2I},k}^*
   \zeta_{\text{y-I2U},mk,s}
{\left[ {\hat{\bar{\bf g}}}_{\text{I2U},mk} \right]}_s  {\left[{\hat{\bar{\bf g}}}_{\text{I2U},mm} \right]}_s^*
   \right|^2,
   \end{align}
   \begin{align}
    &  \bar c_{k,i} \triangleq \frac{1}{N} \sum\limits_{m=1}^{K} \left(M \frac{\beta_{\text{B2I2U},mk} }{v_{\text{I2U},mk}} \left|{\bf a}^T_{\text{B2I},m}{\bf a}_{\text{B2I},i}^* \right|^2+
   MN\frac{\beta_{\text{B2I2U},mk}}{v_{\text{B2I},m}v_{\text{I2U},mk} }
   +MN \frac{\beta_{\text{B2I2U},mk}}{v_{\text{B2I},m}}\right)
   \\
  &+\frac{1}{N}\sum\limits_{m=1}^{K}
  \sum\limits_{s=1}^{M}   \sum\limits_{l=1}^{M} \beta_{\text{B2I2U},mk}
  \left| {\bf a}^T_{\text{B2I},m} {\bf a}_{\text{B2I},i}^* \right|^2
 \zeta_{\text{y-I2U},mk,sl}
{\left[ {\hat{\bar{\bf g}}}_{\text{I2U},mk} \right]}_s  {\left[{\hat{\bar{\bf g}}}_{\text{I2U},mm} \right]}_s^*
{\left[ {\hat{\bar{\bf g}}}_{\text{I2U},mk} \right]}_l^*  {\left[{\hat{\bar{\bf g}}}_{\text{I2U},mm} \right]}_l
\nonumber \\
&+\frac{1}{N}\sum\limits_{m=1}^{K} \sum\limits_{n\ne m}^{K}
   \sum\limits_{s=1}^{M}   \sum\limits_{l=1}^{M}
   {\bf a}^T_{\text{B2I},m} {\bf a}_{\text{B2I},i}^* \sqrt{\beta_{\text{B2I2U},mk}}  {\bf a}^H_{\text{B2I},n}  {\bf a}_{\text{B2I},i} \sqrt{\beta_{\text{B2I2U},nk}}
\nonumber \\
&\times
\zeta_{\text{y-I2U},mk,nk,sl}
{\left[ {\hat{\bar{\bf g}}}_{\text{I2U},mk} \right]}_s  {\left[{\hat{\bar{\bf g}}}_{\text{I2U},mm} \right]}_s^*
{\left[ {\hat{\bar{\bf g}}}_{\text{I2U},nk} \right]}_l^*  {\left[{\hat{\bar{\bf g}}}_{\text{I2U},nn} \right]}_l \nonumber,  \\
 &\bar b_k \triangleq \bar c_{k,k}- \bar a_k.
\end{align}

Furthermore, we can rewrite the optimization problem in a concise vector form:
\begin{align} {\label{E13}}
  \begin{array}{ll}
    \min\limits_{\left\{ {\bf p} \right\}}
& {\bf 1}_{K}^T {\bf p},  \\
 \text{s.t.} & \begin{array}[t]{lll}
            \bar{\bf d}_k^T {\bf p} +
            \left(2^{\bar R_k}-1 \right)  \sigma_0^2 \le 0, k=1,...,K,\\
             {\bf p} \ge {\bf 0},
           \end{array}
  \end{array}
\end{align}
where we define  ${\bf 1}_{K} \triangleq [1,1,...,1]^T \in {\mathbb{C}}^{K \times 1}$ and
$\bar {\bf d}_k \triangleq [\bar d_{k,1},...,\bar d_{k,i},...\bar d_{k,K}]^T$ with $\bar d_{k,i} \triangleq  \left(2^{\bar R_k}-1 \right) \bar c_{k,i}$ for $i \ne k$ and $\bar d_{k,k} \triangleq  \left(2^{\bar R_k}-1 \right) \bar b_k-\bar a_k $.

As such, the original optimization problem becomes a standard linear programming problem, which can be solved by using off-the-shelf optimization tools such as CVX.

It worth noting that the power control problem would become infeasible, if the desired individual rate exceeds a certain threshold. Moreover, as the user location uncertainty (measured by $\Upsilon$) increases, the threshold becomes smaller, due to increased interference.

\begin{remark}
The key advantage of the proposed power control algorithm lies in the fact that it only requires statistical CSI (location information and large fading coefficients), thereby significantly reducing the overhead for CSI acquisition, which is substantially different from conventional power control algorithms requiring instantaneous CSI.\end{remark}

\section{Simulation Results} \label{s5}
In this section, we  provide simulation results to
illustrate the performance
of the location information aided multi-IRS system, as well as to verify  performance of the proposed power control scheme.
The considered system is assumed to operate with the bandwidth of 180 kHz and noise spectral power density of -169 dBm/Hz.
For the BS-IRS channel, the large-scale fading coefficient is modeled by $\alpha_{\text{B2I},m}= C_0 (\frac{d_{\text{B2I},m}}{D_0})^{-\kappa_\text{B2I}}$, where $d_{\text{B2I},m}$ is the distance between the BS and the $m$-th IRS, $C_0$ is the path loss at the reference distance $D_0=1$ meter, and $\kappa_\text{B2I}$ denotes the path loss exponent.
Similarly, for the IRS-user channel, the large-scale fading coefficient is characterized by $\alpha_{\text{I2U},mk}= C_0 (\frac{d_{\text{I2U},mk}}{D_0})^{-\kappa_{\text{I2U}}}$.
Unless otherwise specified, the setup given by Table I is used.
All numerical results are obtained by averaging over 1000 independent small-scale fading parameters for each realization of  user location errors.

\begin{table}
\centering
\label{table2}
\begin{tabular} {|c|c|}
 \hline
 {\bf Parameter} & {\bf Value} \\
 \hline
 Path loss  & $C_0=-30 \text{dB}$ \\
 \hline
Transmit power  & $\rho_d=30 \text{dBm}$ \\
 \hline
Path loss exponent  & $\kappa_\text{B2I}=\kappa_\text{I2U}=2.5$\\
 \hline
Rician K-factor &  $v_\text{I2U}=v_\text{B2I}=v=5$ \\
 \hline
Number of BS antennas & $N=5$\\
 \hline
 Number of   users   & $ K=4$\\
 \hline
 Reflecting element number   & $ M=16$\\
 \hline
BS location &  $(0, 0,0)$
 \\
 \hline
 IRS locations &  $(240, 178, -20), (333, 68, -20),  (362,-75,-20),  (319, -241, -20)$
 \\
 \hline
 User locations &   $(224, 168, -40 ), (314, 64, -40), (343, -71, -40), (303,  -229, -40 )$
 \\
 \hline
\end{tabular}
\caption{Parameter table}
\end{table}

Fig. \ref{Rate_verify} illustrates the achievable sum rate of the system with different user location errors and different Rician K-factors,
where the analytical curves are generated according to Theorem \ref{t2}. As can be readily observed, the numerical results match exactly with the analytical results, thereby validating the correctness of the analytical expressions. Moreover, the sum rate saturates at the high SNRs due to the joint effect of inter-user interference and leakage power. Furthermore, as the user location uncertainty becomes larger, the sum rate decreases significantly.
For instance, the sum rate  with $\rho_d=40 \text{dBm},\Upsilon=0.5$ and $v_\text{B2I}=v_\text{I2U}=5$ is about $16$ bits/s/Hz, but reduces to $5.5$ bits/s/Hz as $\Upsilon$ increases to $2$. In addition, the achievable sum rate increases with the Rician K-factor, due to less interference caused by NLOS paths.
\begin{figure}[!ht]
  \centering
  \includegraphics[width=3.5in]{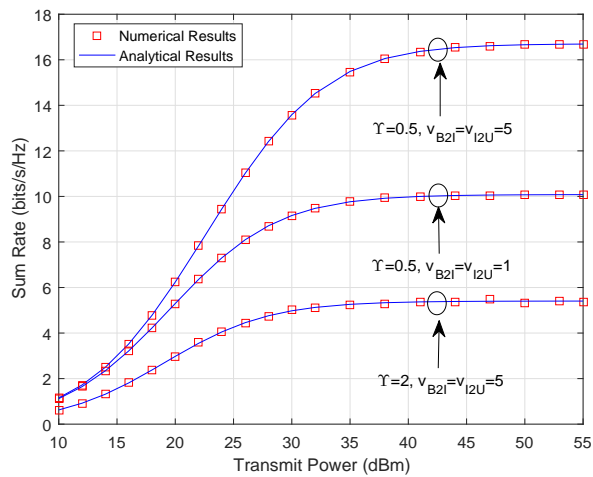}
   \caption{Sum rate of the location information-aided multi-IRS system.}
  \label{Rate_verify}
\end{figure}

Fig. \ref{IRS_direction} shows the impact of IRS directions on the sum rate.
For the ``non-orthogonal IRS directions'' case, the locations of 4 IRSs are given by
 $(278, 113, -20), (338, 41, -20 ), (367, -45, -20)$ and $(370, -151, -20)$, respectively.
For the ``orthogonal IRS directions'' case, the locations of 4 IRSs are given by
 $(240, 178, -20), (333, 68, -20), (362, -75, -20)$ and $(319, -241, -20)$, respectively.
The analytical curve with orthogonal IRS directions (  ${\bf a}_{\text{B2I},m}^T {\bf a}_{\text{B2I},i}^* = 0, m \ne i$ ) is plotted according to Corollary \ref{c2}, while the analytical curve with non-orthogonal IRS directions is obtained according to Theorem \ref{t2}.
We can see that the analytical results with orthogonal IRS directions well match their numerical results, which verifies our analytical results in Corollary \ref{c2}.
Moreover, as expected in Proposition \ref{p2} , the achievable sum rate with orthogonal IRS directions is much higher than that with non-orthogonal IRS directions, due to the decreased  interference from multiple IRSs.

\begin{figure}[!ht]
  \centering
  \includegraphics[width=3.5in]{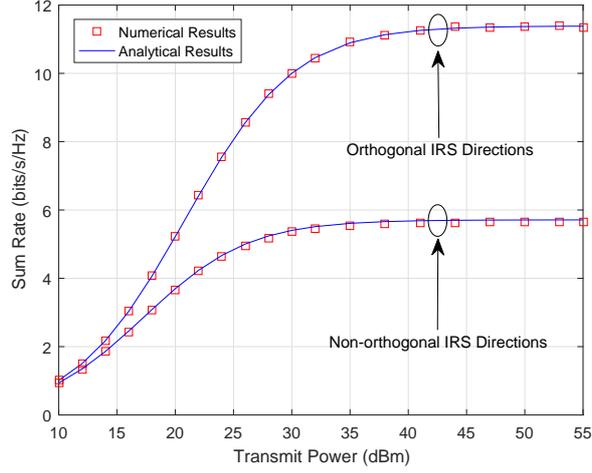}
   \caption{The impact of IRS directions on the sum rate with $\Upsilon=1$m.}
  \label{IRS_direction}
\end{figure}



Fig. \ref{IRS_location} shows the impact of the IRS location on the sum rate, when the locations of users are fixed. As can be readily observed, the sum rate is not a monotonic function with respect to the BS-IRS distances. The worst sum rate occurs when the IRSs are deployed in the middle of the BS and users, and the sum rate improves when moving the IRSs towards the BS or users, which indicates that it is desirable to deploy the IRSs near the BS or users.
\begin{figure}[!ht]
  \centering
  \includegraphics[width=3.5in]{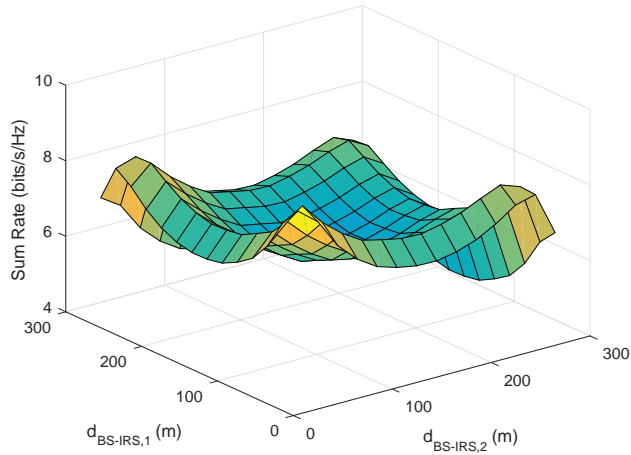}
   \caption{ The impact of the IRS location on the sum rate.}
  \label{IRS_location}
\end{figure}

Fig. \ref{Rate_VS_Upsilon} shows the  impact of the number of reflecting elements on the sum rate, where the  approximate results are obtained according to Corollary \ref{c3}. We can see that the approximate results matches the numerical results well, thereby validating the correctness of Corollary \ref{c3}.
Moreover, as expected in  Corollary \ref{c3}, the sum rate grows logarithmically with the number of reflecting elements, indicating the great benefit of applying a large number of reflecting elements.
In addition, increasing the antenna number can also improve the sum rate due to the increased BS beam gain.

\begin{figure}[!ht]
  \centering
  \includegraphics[width=3.5in]{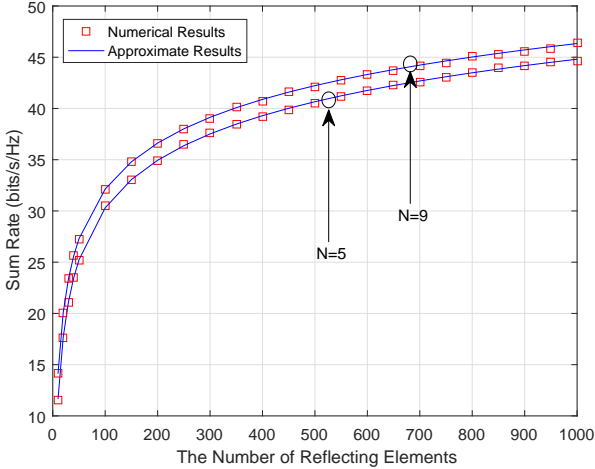}
  \caption{ The impact of the number of reflecting elements on the sum rate.}
  \label{Rate_VS_Upsilon}
\end{figure}

Fig. \ref{Antenna_number} shows the  impact of the antenna number on the sum rate, where the ``Limit'' curve is  generated according to Corollary \ref{c4}.
As expected, the sum rate increases with the number of antennas and gradually converges to a limit, which is mainly determined by the number of reflecting elements and the Rician K-factor of IRS-user channels.
It can be seen that the sum rate limit increases with the number reflecting elements, while decreases
as the the Rician K-factor of IRS-user channels becomes smaller due to more interference caused by NLOS paths.

\begin{figure}[!ht]
  \centering
  \includegraphics[width=3.5in]{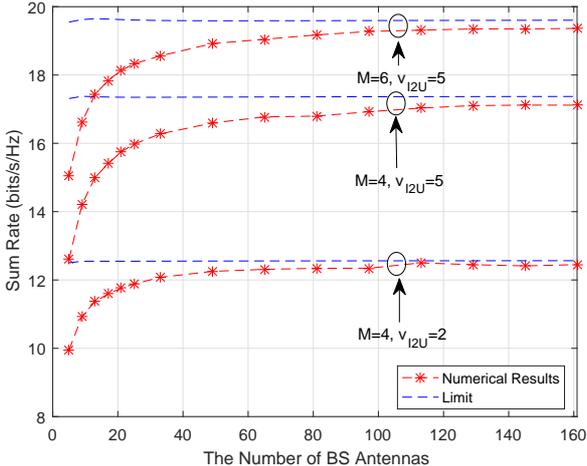}
  \caption{The impact of the number of BS antennas on the sum rate with perfect location information and $\rho_d= 40 \text{dBm}$.}
  \label{Antenna_number}
\end{figure}

Fig. \ref{NLOS_paths} illustrates the  impact of the Rician K-factor on the sum rate, where the ``Limit'' curve is  generated according to Corollary \ref{c5}.
It can be seen that the sum rate grows with the Rician K-factor and gradually converges to a limit, which is mainly determined by the number of reflecting elements.
As expected in Corollary \ref{c5},
increasing the number of reflecting elements greatly improves the sum rate, due to both the phase shift beamforming gain and the inherent aperture gain of the IRS.
For instance, with a large Rician K-factor, the sum rate increases significantly from about 5 bits/s/Hz to about 11 bits/s/Hz, as the number of reflecting elements grows from 4 to 8.

\begin{figure}[!ht]
  \centering
  \includegraphics[width=3.5in]{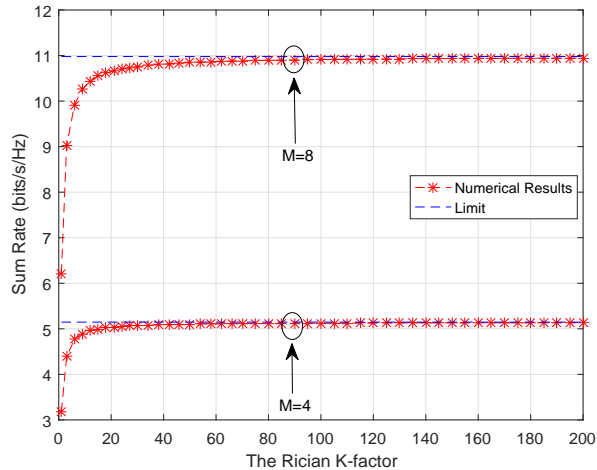}
  \caption{The impact of the Rician K-factor on the sum rate with perfect location information.}
  \label{NLOS_paths}
\end{figure}



 Fig. \ref{beamforming_comparison} shows the performance of  the proposed beamforming scheme, where the optimal power control  scheme discussed in Section \ref{PC} is performed, and we assume $\bar R_k=\bar R, k=1,...,K$ with $\bar R$ being the desired individual rate.
For comparison, the joint optimization algorithm in \cite{wu2019intelligent} is presented as the benchmark, where the multiple IRSs are treated as a big distributed IRS.
It can be observed that, when the desired rate is not very high, the
proposed scheme is superior to the benchmark scheme. This is reasonable because with small rate
constraint, the required transmit power is not high, as such, the system is likely to be noise-limited. In a noise-limited scenario, the MRC-based scheme is near optimal. Then, combining with the optimal power control algorithm eventually yields superior performance. However, as the rate constraint becomes
stringent, the required transmit power grows and the system becomes interference-limited. Then,
the performance of the proposed MRC scheme deteriorates significantly due to severe interference,
and becomes inferior to the benchmark scheme. In addition, it may become infeasible if the
desired rate is greater than a threshold. The reason is that the proposed beamforming scheme
is designed without considering interference suppression, which is only partially addressed by
power control and proper IRS configurations, i.e., deploying multiple IRSs at distinct directions,
increasing the distance between different IRSs, and assigning each user to its nearest IRS. In
contrast, the benchmark scheme assumes perfect CSI and adopts a joint optimization method,
where the transmit beam and the phase shift beam are optimized by solving two semidefinite
programming (SDP) sub-problems iteratively. Hence, the benchmark scheme can better deal with
the inter-user interference.

\begin{figure}[!ht]
  \centering
  \includegraphics[width=3.5in]{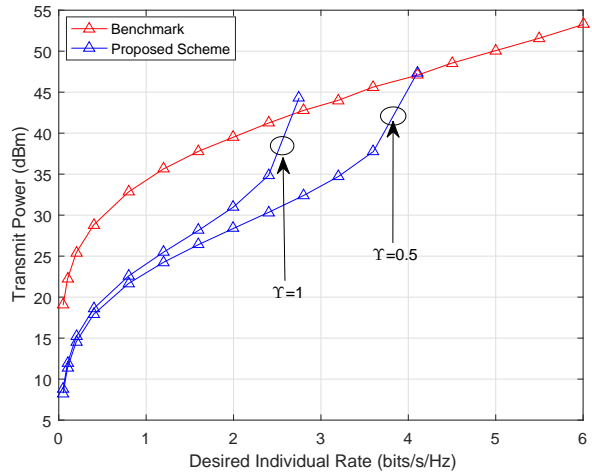}
  \caption{ The performance of the proposed beamforming scheme .}
  \label{beamforming_comparison}
\end{figure}

\section{Conclusion} \label{s6}
This paper has proposed a location information aided  multi-IRS system. Based on the imperfect GPS location information, the effective angles are estimated, and the statistics of the angle estimation error are derived. Then, exploiting the estimated angles, low-complexity transmit beam and IRS phase shift beam have been designed, and closed-form expressions are derived for the achievable rate, which facilitates the characterization of the benefits of deploying a large number of BS antennas or reflecting elements. The findings suggest that the IRS-aided link can obtain a power gain of order $N M^2$. With a large number of reflecting elements, the individual SINR is proportional to $M$, while becomes proportional to $M^2$ when there are no NLOS paths. Besides, it is demonstrated that the achievable rate degrades significantly as user location uncertainty increases. Furthermore, the analytical findings suggest that IRSs should be deployed at distinct directions and be far away from each other to reduce the interference from multiple IRSs. Finally,  based on the proposed beamforming scheme, an optimal power control scheme is proposed to minimize the total transmit power under individual rate constraints.
For future works, the robust beamforming in the presence of location  uncertainty can be investigated, to guarantee the worse-case quality of service requirement \cite{J.Zhang,X.Hu1}. Besides, when the number of users is greater than the number of IRSs, how to schedule users in different resource blocks and how to pair the user with the IRS are worth studying.

\begin{appendices}
\section{Proof of Proposition \ref{p1}} \label{Ax1}
The effective AOD from the $m$-th IRS  to the $k$-th  user is given by
\begin{align}
 \vartheta_{\text{y-I2U},mk} =\frac{y_{\text{I},m}- y_{\text{U},k}}{ d_{\text{I2U},mk}}.
\end{align}

Using (\ref{E1}), the above equation can be rewritten as
\begin{align} \label{E2}
   \vartheta_{\text{y-I2U},mk} =\frac{\hat d_{\text{I2U},mk}}{d_{\text{I2U},mk}}  \hat  \vartheta_{\text{y-I2U},mk}-\frac{
   \Delta y_{\text{U},k}}{ d_{\text{I2U},mk}},
\end{align}
where $d_{\text{I2U},mk}=\sqrt{(z_{\text{I},m}- z_{\text{U},k})^2+(y_{\text{I},m}- y_{\text{U},k})^2+(x_{\text{I},m}- x_{\text{U},k})^2}$ is the  distance between the BS and the $k$-th user, and $\Delta y_{\text{U},k}=y_{\text{U},k}-\hat y_{\text{U},k}$ is the location error along the $y$ axis.

Next, we focus on the term $\frac{\hat d_{\text{I2U},mk}}{d_{\text{I2U},mk}}$:
\begin{align}
  \frac{\hat d_{\text{I2U},mk}}{d_{\text{I2U},mk}}=  {\left( 1+Q\right)}^{-\frac{1}{2} },
\end{align}
where
\begin{align}
 Q & \triangleq \frac{{\Delta z_{\text{U},k}}^2 +{\Delta y_{\text{U},k}}^2 +{\Delta x_{\text{U},k}}^2 }{{\hat d}^2_{\text{I2U},mk}}  -2\left( \frac{  \hat\vartheta_{\text{z-I2U},mk} \Delta z_{\text{U},k} + \hat\vartheta_{\text{y-I2U},mk} \Delta y_{\text{U},k}+\hat\vartheta_{\text{x-I2U},mk} \Delta x_{\text{U},k} }{\hat d_{\text{I2U},mk}}\right) , \\
 & \approx -2\left( \frac{  \hat\vartheta_{\text{z-I2U},mk} \Delta z_{\text{U},k} + \hat\vartheta_{\text{y-I2U},mk} \Delta y_{\text{U},k}+\hat\vartheta_{\text{x-I2U},mk} \Delta x_{\text{U},k} }{\hat d_{\text{I2U},mk}}\right),\nonumber
\end{align}
where $\hat \vartheta_{\text{x-I2U},mk}  \triangleq \frac{x_{\text{I},m}-\hat x_{\text{U},k}}{\hat d_{\text{I2U},mk}}$, $\hat \vartheta_{\text{z-I2U},mk}  \triangleq \frac{z_{\text{I},m}-\hat z_{\text{U},k}}{\hat d_{\text{I2U},mk}}$, $\Delta z_{\text{U},k}=z_{\text{U},k}-\hat z_{\text{U},k}$ and $\Delta x_{\text{U},k}=x_{\text{U},k}-\hat x_{\text{U},k}$  are location errors along $z$ and $x$ axes, respectively.

Using the Taylor expansion of  ${\left( 1+Q\right)}^{-\frac{1}{2} }$ at $Q=0$, we have
\begin{align} \label{E3}
  \frac{\hat d_{\text{I2U},mk}}{d_{\text{I2U},mk}}=  1-\frac{1}{2}Q.
\end{align}

Substituting (\ref{E3}) into (\ref{E2}), we obtain
\begin{align}
   \vartheta_{\text{y-I2U},mk} =
   \left(1-\frac{1}{2}Q \right) \hat  \vartheta_{\text{y-I2U},mk}
   -\left(1-\frac{1}{2}Q \right) \frac{\Delta y_{\text{U},k}}{ \hat d_{\text{I2U},mk}}
   \approx  \hat\vartheta_{\text{y-I2U},mk} +\epsilon_{\text{y-I2U},mk},
\end{align}
where the estimation error of the effective AOD along $y$ axis, $\epsilon_{\text{y-I2U},mk}$, is given by
\begin{align}
 \epsilon_{\text{y-I2U},mk} &=
  - \frac{1}{2}Q\hat\vartheta_{\text{y-I2U},mk}-
  \frac{\Delta y_{\text{U},k}}{ \hat d_{\text{I2U},mk}} ,\\
  &=  \frac{ \left( \hat\vartheta^2_{\text{y-I2U},mk}-1\right) \Delta y_{\text{U},k} + \hat\vartheta_{\text{y-I2U},mk} \hat\vartheta_{\text{z-I2U},mk} \Delta z_{\text{U},k}+\hat\vartheta_{\text{y-I2U},mk}\hat\vartheta_{\text{x-I2U},mk} \Delta x_{\text{U},k} }{\hat d_{\text{I2U},mk}} .\nonumber
\end{align}
To this end, we complete the proof of (\ref{E5}).

\section{Proof of Theorem \ref{t1}} \label{At1}
The equation (\ref{E5}) given  by Proposition \ref{p1} can be rewritten as
 \begin{align}
    \epsilon_{\text{y-I2U},mk} = a_{\text{y-I2U},mk} \Delta y_{\text{U},k}+ b_{\text{y-I2U},mk} \Delta z_{\text{U},k}
    + c_{\text{y-I2U},mk} \Delta x_{\text{U},k}
 \end{align}
 where $a_{\text{y-I2U},mk} \triangleq \frac{ \hat\vartheta^2_{\text{y-I2U},mk}-1}{\hat d_{\text{I2U},mk}}$, $b_{\text{y-I2U},mk} \triangleq \frac{\hat\vartheta_{\text{y-I2U},mk} \hat\vartheta_{\text{z-I2U},mk}}{\hat d_{\text{I2U},mk}}$, and $c_{\text{y-I2U},mk} \triangleq \frac{\hat\vartheta_{\text{y-I2U},mk}\hat\vartheta_{\text{x-I2U},mk} }{\hat d_{\text{I2U},mk}}$.

 The cumulative distribution function (CDF) of $\epsilon_{\text{y-I2U},mk}$ is given by
 \begin{align}
     F_{\epsilon_{\text{y-I2U},mk}} \left( x\right)&= P \left(\epsilon_{\text{y-I2U},mk} < x \right)\\
     &=P \left(a_{\text{y-I2U},mk} \Delta y_{\text{U},k}+ b_{\text{y-I2U},mk} \Delta z_{\text{U},k}
    + c_{\text{y-I2U},mk} \Delta x_{\text{U},k} - x<0 \right).\nonumber
 \end{align}

 Recall that the point $(\Delta x_{\text{U},k}, \Delta y_{\text{U},k}, \Delta z_{\text{U},k})$ is uniformly distributed within a sphere with the radius $\Upsilon$ and the origin $(0,0,0)$.  The distance from the origin $(0,0,0)$ to the plane $a_{\text{y-I2U},mk} \Delta y_{\text{U},k}+ b_{\text{y-I2U},mk} \Delta z_{\text{U},k} + c_{\text{y-I2U},mk} \Delta x_{\text{U},k} - x=0$  is
 \begin{align} \label{E6}
     d=\frac{\left|x\right|}{\sqrt{ a^2_{\text{y-I2U},mk} +b^2_{\text{y-I2U},mk}+c^2_{\text{y-I2U},mk} }}.
 \end{align}

 For  $-\Upsilon \sqrt{ a^2_{\text{y-I2U},mk} +b^2_{\text{y-I2U},mk}+c^2_{\text{y-I2U},mk} }\le x \le 0$, we have
 \begin{align} \label{E7}
  & F_{\epsilon_{\text{y-I2U},mk}} \left( x\right) =1/\left( \frac{4 \pi \Upsilon^3}{3}  \right)\int_{d}^\Upsilon \pi \left(\Upsilon^2 -t^2\right)dt \\
    &=1/\left( \frac{4 \pi \Upsilon^3}{3}  \right) \left(\frac{\pi}{3} d^3-\pi \Upsilon^2 d+\frac{2 \pi}{3} \Upsilon^3\right) .\nonumber
 \end{align}

  For  $0<x\le \Upsilon \sqrt{ a^2_{\text{y-I2U},mk} +b^2_{\text{y-I2U},mk}+c^2_{\text{y-I2U},mk} }$, we have
 \begin{align}\label{E8}
  & F_{\epsilon_{\text{y-I2U},mk}} \left( x\right) =1-1/\left( \frac{4 \pi \Upsilon^3}{3}  \right)\int_{d}^\Upsilon \pi \left(\Upsilon^2 -t^2\right)dt \\
    &=1-1/\left( \frac{4 \pi \Upsilon^3}{3}  \right) \left(\frac{\pi}{3} d^3-\pi \Upsilon^2 d+\frac{2 \pi}{3} \Upsilon^3\right). \nonumber
 \end{align}

 Substituting (\ref{E6}) into (\ref{E7}) and (\ref{E8}), we obtain
 \begin{align}
 F_{\epsilon_{\text{y-I2U},mk}} \left( x\right)  &=\frac{1}{2} - \frac{1}{4 \Upsilon^3} {\left( a^2_{\text{y-I2U},mk} +b^2_{\text{y-I2U},mk}+c^2_{\text{y-I2U},mk}  \right)}^{-3/2} x^3 \\
 & +\frac{3}{4 \Upsilon} {\left( a^2_{\text{y-I2U},mk} +b^2_{\text{y-I2U},mk}+c^2_{\text{y-I2U},mk}  \right)}^{-1/2} x,
  \nonumber
 \end{align}
 for $ \left|x \right|\le \Upsilon \sqrt{ a^2_{\text{y-I2U},mk} +b^2_{\text{y-I2U},mk}+c^2_{\text{y-I2U},mk} }$.

 Define
$
 \Phi_{\text{y-I2U},mk} \triangleq \sqrt {{\left( \hat\vartheta^2_{\text{y-I2U},mk}-1 \right)}^2+\hat\vartheta^2_{\text{y-I2U},mk} \hat\vartheta^2_{\text{z-I2U},mk}+\hat\vartheta^2_{\text{y-I2U},mk} \hat\vartheta^2_{\text{x-I2U},mk}}
$. The corresponding PDF is given by
 \begin{align}
 f_{\epsilon_{\text{y-I2U},mk}} \left( x\right) = \frac{\partial F_{\epsilon_{\text{y-I2U},mk}} \left( x\right) }{\partial x}
 = - \frac{3 {\hat{d}}^3_{\text{I2U},mk} } {4 \Upsilon^3} {\Phi}^{-3}_{\text{y-I2U},mk} x^2  +\frac{3 {\hat{d}}_{\text{I2U},mk}}{4 \Upsilon} {\Phi}^{-1}_{\text{y-I2U},mk},
  \nonumber
 \end{align}
 for $|x|\le \frac{\Upsilon}{\hat{d}_{\text{I2U},mk}} {\Phi}_{\text{y-I2U},mk} $.

 Based on the above PDF, the mean and the variance of $\epsilon_{\text{y-I2U},mk}$ can be obtained.

\section{Proof of Lemma \ref{L2}} \label{AL2}
Denote $\epsilon_{k,mn,sl} \triangleq \left(s-1\right)\epsilon_{\text{y-I2U},mk}
 -\left(l-1\right)\epsilon_{\text{y-I2U},nk}$.
 \begin{align}
 \zeta_{\text{y-I2U},k,mn,sl}
 \triangleq
  \mathbb{E}\left\{
   \varrho_{\text{I2U},mk,s}
  \varrho_{\text{I2U},nk,l}^* \right\}
= \mathbb{E}\left\{
 e^{j\pi \epsilon_{k,mn,sl}}
 \right\}.
 \end{align}

 We can rewrite $\epsilon_{k,mn,sl}$ as
 \begin{align}
     \epsilon_{k,mn,sl}= a_{k,mn,sl} \Delta y_{\text{U},k}+ b_{k,mn,sl} \Delta z_{\text{U},k}
    + c_{k,mn,sl} \Delta x_{\text{U},k},
 \end{align}
 where
 \begin{align}
  & a_{k,mn,sl} \triangleq
   \left(s-1\right) \frac{ \hat\vartheta^2_{\text{y-I2U},mk}-1}{\hat d_{\text{I2U},mk}} - \left(l-1\right) \frac{ \hat\vartheta^2_{\text{y-I2U},nk}-1}{\hat d_{\text{I2U},nk}}, \\
  & b_{k,mn,sl} \triangleq
    \left(s-1\right)\frac{\hat\vartheta_{\text{y-I2U},mk} \hat\vartheta_{\text{z-I2U},mk}}{\hat d_{\text{I2U},mk}}- \left(l-1\right)\frac{\hat\vartheta_{\text{y-I2U},nk} \hat\vartheta_{\text{z-I2U},nk}}{\hat d_{\text{I2U},nk}},\\
    &c_{k,mn,sl}\triangleq
    \left(s-1\right)  \frac{\hat\vartheta_{\text{y-I2U},mk}\hat\vartheta_{\text{x-I2U},mk} }{\hat d_{\text{I2U},mk}}- \left(l-1\right)  \frac{\hat\vartheta_{\text{y-I2U},nk}\hat\vartheta_{\text{x-I2U},nk} }{\hat d_{\text{I2U},nk}}.
 \end{align}

 Following the similar process of the proof of Theorem \ref{t1}, we have the PDF of $\epsilon_{k,mn,sl}$ given by
  \begin{align}
  & f_{\epsilon_{k,mn,sl}} \left( x\right) \\
&= -\frac{3}{4 \Upsilon^3} {\left( a^2_{k,mn,sl} +b^2_{k,mn,sl}+c^2_{k,mn,sl}  \right)}^{-3/2} x^2  +\frac{3}{4 \Upsilon} {\left( a^2_{k,mn,sl} +b^2_{k,mn,sl}+c^2_{k,mn,sl}  \right)}^{-1/2}\nonumber\\
&=-\frac{3    }{4 \Upsilon^3 \Phi_{k,mn,sl}^3} x^2
+\frac{3  }{4 \Upsilon \Phi_{k,mn,sl}},\nonumber
 \end{align}
 for $|x|\le \Upsilon {\Phi}_{k,mn,sl} $,
 where we define
 \begin{align}
 \Phi_{k,mn,sl} \triangleq \sqrt { a^2_{k,mn,sl} +b^2_{k,mn,sl}+c^2_{k,mn,sl} }.
 \end{align}

According to the above PDF, the expectation $\zeta_{\text{y-I2U},k,mn,sl}
= \mathbb{E}\left\{
 e^{j\pi \epsilon_{k,mn,sl}}
 \right\}$ can be calculated as
 \begin{align}
& \zeta_{k,mn,sl}
= \mathbb{E}\left\{
 e^{j\pi \epsilon_{k,mn,sl}}
 \right\}\\
 &=\begin{cases}
1, & s=1\ \text{and} \  l=1\\
 \frac{3}{\varpi_{ k,mn,sl}^2} \left(
\frac{\sin \varpi_{ k,mn,sl} }{\varpi_{ k,mn,sl}} -\cos\varpi_{ k,mn,sl} \right), \nonumber & \text{else}
\end{cases}
 \end{align}
where
$
    \varpi_{ k,mn,sl} \triangleq {\pi {\Phi}_{k,mn,sl} \Upsilon}
$.

\section{Proof of Theorem \ref{t2}} \label{At2}
\subsection{ Calculate $A_k$} \label{AA}

Recall that ${\bf w}_{k}=\sqrt{\frac{\eta_k \rho_d}{N  } } {\bf a}_{\text{B2I},k}^*$. We have
\begin{align}
   & \mathbb{E}\left\{  {\bf g}_{k}^T   {\bf w}_k \right\}
    =  \sqrt{\frac{\eta_i \rho_d}{N } } \mathbb{E}\left\{  {\bf g}_{k}^T {\bf a}_{\text{B2I},k}^*   \right\} \\
   &  =  \sqrt{\frac{\eta_k \rho_d}{N } } \mathbb{E}\left\{  {\bf g}_{\text{I2U},kk}^T { \bm \Theta_{k}} {\bf G}_{\text{B2I},k} {\bf a}_{\text{B2I},k}^*  \right\}
   +\sqrt{\frac{\eta_k \rho_d}{N } } \sum\limits_{m\ne k}^{K} \mathbb{E}\left\{  {\bf g}_{\text{I2U},mk}^T { \bm \Theta_{m}} {\bf G}_{\text{B2I},m} {\bf a}_{\text{B2I},k}^*  \right\}. \nonumber
    \end{align}

We first calculate $ \mathbb{E}\left\{  {\bf g}_{\text{I2U},mk}^T { \bm \Theta_{m}} {\bf G}_{\text{B2I},m} {\bf a}_{\text{B2I},k}^*  \right\}, m\ne k$:
\begin{align}
  &\mathbb{E}\left\{  {\bf g}_{\text{I2U},mk}^T { \bm \Theta_{m}} {\bf G}_{\text{B2I},m} {\bf a}_{\text{B2I},k}^*  \right\}=
   \sqrt{\beta_{\text{B2I2U},mk}} \mathbb{E}\left\{  \bar{\bf g}_{\text{I2U},mk}^T { \bm \Theta_{m}}  {\bf b}_{\text{B2I},m}  \right\}
   {\bf a}^T_{\text{B2I},m} {\bf a}_{\text{B2I},k}^*\\
   &=  {\bf a}^T_{\text{B2I},m} {\bf a}_{\text{B2I},k}^* \sqrt{\beta_{\text{B2I2U},mk}} \mathbb{E}\left\{ {\bm \xi}^T_m \left( \bar{\bf g}_{\text{I2U},mk} \odot  {\bf b}_{\text{B2I},m}\right) \right\} \nonumber \\
   &= {\bf a}^T_{\text{B2I},m} {\bf a}_{\text{B2I},k}^* \sqrt{\beta_{\text{B2I2U},mk}}
   \sum\limits_{s=1}^{M}
\mathbb{E}\left\{  \varrho_{\text{I2U},mk,s} \right\}
{\left[ {\hat{\bar{\bf g}}}_{\text{I2U},mk} \right]}_s  {\left[{\hat{\bar{\bf g}}}_{\text{I2U},mm} \right]}_s^*
\nonumber.
\end{align}

Denoting $\zeta_{\text{y-I2U},mk,s}\triangleq  \mathbb{E}\left\{ \varrho_{\text{I2U},mk,s} \right\}$ and using the results given by Theorem \ref{t1}, we have
\begin{align} \label{E10}
 \zeta_{\text{y-I2U},mk,s}
& =\int_{-\frac{\Upsilon {\Phi}_{\text{y-I2U},mk}}{\hat{d}_{\text{I2U},mk}}  }^{\frac{\Upsilon {\Phi}_{\text{y-I2U},mk}}{\hat{d}_{\text{I2U},mk}} } \left( - \frac{3 {\hat{d}}^3_{\text{I2U},mk} } {4 \Upsilon^3} {\Phi}^{-3}_{\text{y-I2U},mk} x^2  +\frac{3 {\hat{d}}_{\text{I2U},mk}}{4 \Upsilon} {\Phi}^{-1}_{\text{y-I2U},mk} \right)
  e^{ j\pi
\left(s-1\right) x } \ dx \\
& =
\begin{cases}
 \frac{3}{\varpi_{ \text{y-I2U},mk,s}^2} \left(
\frac{\sin \varpi_{  \text{y-I2U},mk,s} }{\varpi_{  \text{y-I2U},mk,s}} -\cos\varpi_{  \text{y-I2U},mk,s} \right) & s \!\ne\! 1\\
1& s\! =\!1
\end{cases}.\nonumber
\end{align}

Thus, we obtain
\begin{align} \label{Ak2}
  &\mathbb{E}\left\{  {\bf g}_{\text{I2U},mk}^T { \bm \Theta_{m}} {\bf G}_{\text{B2I},m} {\bf a}_{\text{B2I},k}^*  \right\} \\
&={\bf a}^T_{\text{B2I},m} {\bf a}_{\text{B2I},k}^* \sqrt{\beta_{\text{B2I2U},mk}}
  \sum\limits_{s=1}^{M} \zeta_{\text{y-I2U},mk,s}
{\left[ {\hat{\bar{\bf g}}}_{\text{I2U},mk} \right]}_s  {\left[{\hat{\bar{\bf g}}}_{\text{I2U},mm} \right]}_s^*
\nonumber.
\end{align}

 Then, we start the calculation of the first term:
    \begin{align}
       &\mathbb{E}\left\{  {\bf g}_{\text{I2U},kk}^T { \bm \Theta_{k}} {\bf G}_{\text{B2I},k} {\bf a}_{\text{B2I},k}^*  \right\}
       =N \sqrt{\beta_{\text{B2I2U},kk}} \mathbb{E}\left\{  \bar{\bf g}_{\text{I2U},kk}^T { \bm \Theta_{k}} {\bf b}_{\text{B2I},k} \right\}\\
       &=N \sqrt{\beta_{\text{B2I2U},kk}} \mathbb{E}\left\{ {\bm \xi}^T  \left( \bar{\bf g}_{\text{I2U},kk} \odot  {\bf b}_{\text{B2I},k}\right) \right\},\nonumber
    \end{align}
where $\beta_{\text{B2I2U},mk} \triangleq \frac{\alpha_{\text{B2I},m} v_{\text{B2I},m} \alpha_{\text{I2U},mk} v_{\text{I2U},mk} }{  (v_{\text{B2I},m}+1 ) (v_{\text{I2U},mk}+1 )}$.

Recall that ${\bm \xi}_k=\left({\hat{\bar{\bf g}}}_{\text{I2U},kk} \odot {\bf b}_{\text{B2I},k} \right)^*$. The above equation can be expressed as
\begin{align}
\mathbb{E}\left\{  {\bf g}_{\text{I2U},kk}^T { \bm \Theta_{k}} {\bf G}_{\text{B2I},k} {\bf a}_{\text{B2I},k}^*  \right\}
= N \sqrt{\beta_{\text{B2I2U},kk}} \sum\limits_{s=1}^{M}
\mathbb{E}\left\{  \varrho_{\text{I2U},kk,s} \right\}.
\end{align}

As such, we obtain
\begin{align} \label{Ak1}
\mathbb{E}\left\{  {\bf g}_{\text{I2U},kk}^T { \bm \Theta_{k}} {\bf G}_{\text{B2I},k} {\bf a}_{\text{B2I},k}^*  \right\}
= N \sqrt{\beta_{\text{B2I2U},kk}} \sum\limits_{s=1}^{M} \zeta_{\text{y-I2U},kk,s}.
\end{align}

Combining \ref{Ak2} and \ref{Ak1} yields
\begin{align}
 & A_k=\frac{\eta_k \rho_d}{N } \left| \sum\limits_{m= 1}^{K}  \sum\limits_{s=1}^{M} \sqrt{\beta_{\text{B2I2U},mk}}
 {\bf a}^T_{\text{B2I},m} {\bf a}_{\text{B2I},k}^*
   \zeta_{\text{y-I2U},mk,s}
{\left[ {\hat{\bar{\bf g}}}_{\text{I2U},mk} \right]}_s  {\left[{\hat{\bar{\bf g}}}_{\text{I2U},mm} \right]}_s^*
   \right|^2.
\end{align}

\subsection{ Calculate $C_{k,i}$} \label{AB}
We can rewrite $C_{k,i}$ as
\begin{align}
     C_{k,i}&=  \mathbb{E}\left\{  {\left| { {\bf g}_{k}^T {\bf w}_i } \right|}^2\right\}
   =\frac{\eta_i \rho_d}{N}\sum\limits_{m=1}^{K}
   \sum\limits_{n=1}^{K}
    \mathbb{E}\left\{
   {\bf g}_{\text{I2U},mk}^T { \bm \Theta_{m}} {\bf G}_{\text{B2I},m} {\bf a}_{\text{B2I},i}^*
   {\bf a}_{\text{B2I},i}^T {\bf G}_{\text{B2I},n}^H
    { \bm \Theta_{n}}^H {\bf g}_{\text{I2U},nk}^*
    \right\}.
\end{align}

1) For $m=n$, we have
\begin{align}
   & \mathbb{E}\left\{
   {\bf g}_{\text{I2U},mk}^T { \bm \Theta_{m}} {\bf G}_{\text{B2I},m} {\bf a}_{\text{B2I},i}^*
   {\bf a}_{\text{B2I},i}^T {\bf G}_{\text{B2I},n}^H
    { \bm \Theta_{n}}^H {\bf g}_{\text{I2U},nk}^*
    \right\}\\
    &=\mathbb{E}\left\{ \left| {\bf g}_{\text{I2U},mk}^T { \bm \Theta_{m}} {\bf G}_{\text{B2I},m} {\bf a}_{\text{B2I},i}  \right|^2\right\}\nonumber\\
    &= M \frac{\beta_{\text{B2I2U},mk} }{v_{\text{I2U},mk}} \left|{\bf a}^T_{\text{B2I},m}{\bf a}_{\text{B2I},i}^* \right|^2+
   MN\frac{\beta_{\text{B2I2U},mk}}{v_{\text{B2I},m}v_{\text{I2U},mk} }
   +MN \frac{\beta_{\text{B2I2U},mk}}{v_{\text{B2I},m}}
   \nonumber \\
   &+\beta_{\text{B2I2U},mk}\mathbb{E}\left\{ \left| \bar{\bf g}_{\text{I2U},mk}^T { \bm \Theta_{m}} \bar{\bf G}_{\text{B2I},m} {\bf a}_{\text{B2I},i}^*  \right|^2\right\}.\nonumber
\end{align}

 Next, we focus on the calculation of $\mathbb{E}\left\{ \left| \bar{\bf g}_{\text{I2U},mk}^T { \bm \Theta_{m}} \bar{\bf G}_{\text{B2I},m} {\bf a}_{\text{B2I},i}^*  \right|^2\right\}$.
 \begin{align}
  &\mathbb{E}\left\{ \left| \bar{\bf g}_{\text{I2U},mk}^T { \bm \Theta_{m}} \bar{\bf G}_{\text{B2I},m} {\bf a}_{\text{B2I},i}^*  \right|^2\right\}
  \\
  &=\left| {\bf a}^T_{\text{B2I},m} {\bf a}_{\text{B2I},i}^* \right|^2
   \sum\limits_{s=1}^{M}   \sum\limits_{l=1}^{M}
 \mathbb{E}\left\{ \varrho_{\text{I2U},mk,sl} \right\}
{\left[ {\hat{\bar{\bf g}}}_{\text{I2U},mk} \right]}_s  {\left[{\hat{\bar{\bf g}}}_{\text{I2U},mm} \right]}_s^*
{\left[ {\hat{\bar{\bf g}}}_{\text{I2U},mk} \right]}_l^*  {\left[{\hat{\bar{\bf g}}}_{\text{I2U},mm} \right]}_l\nonumber,
 \end{align}
 where  $\varrho_{\text{I2U},mk,sl}\triangleq e^{j\pi \left(  s-l\right)\epsilon_{\text{y-I2U},mk}} $.

Following the similar process of the calculation of $  \mathbb{E}\left\{ \varrho_{\text{I2U},mk,s} \right\}$, we have
\begin{align}
   \zeta_{\text{y-I2U},mk,sl}
   \triangleq \mathbb{E}\left\{ \varrho_{\text{I2U},mk,sl} \right\} =
\begin{cases}
\frac{3}{\varpi_{ \text{y-I2U},mk,sl}^2} \left(
\frac{\sin \varpi_{  \text{y-I2U},mk,sl} }{\varpi_{  \text{y-I2U},mk,sl}} -\cos\varpi_{  \text{y-I2U},mk,sl} \right) & s \!\ne\! l\\
1& s\! =\!l
\end{cases}.
\end{align}

Thus, we obtain
\begin{align}
  &\mathbb{E}\left\{ \left| \bar{\bf g}_{\text{I2U},mk}^T { \bm \Theta_{m}} \bar{\bf G}_{\text{B2I},m} {\bf a}_{\text{B2I},i}^* \right|^2\right\}
  \\
  &=\left| {\bf a}^T_{\text{B2I},m} {\bf a}_{\text{B2I},i}^* \right|^2
   \sum\limits_{s=1}^{M}   \sum\limits_{l=1}^{M}
 \zeta_{\text{y-I2U},mk,sl}
{\left[ {\hat{\bar{\bf g}}}_{\text{I2U},mk} \right]}_s  {\left[{\hat{\bar{\bf g}}}_{\text{I2U},mm} \right]}_s^*
{\left[ {\hat{\bar{\bf g}}}_{\text{I2U},mk} \right]}_l^*  {\left[{\hat{\bar{\bf g}}}_{\text{I2U},mm} \right]}_l\nonumber.
 \end{align}

2) For $m\ne n$, we have
\begin{align}
     & \mathbb{E}\left\{
   {\bf g}_{\text{I2U},mk}^T { \bm \Theta_{m}} {\bf G}_{\text{B2I},m} {\bf a}_{\text{B2I},i}^*
   {\bf a}_{\text{B2I},i}^T {\bf G}_{\text{B2I},n}^H
    { \bm \Theta_{n}}^H {\bf g}_{\text{I2U},nk}^*
    \right\}\\
    &={\bf a}^T_{\text{B2I},m} {\bf a}_{\text{B2I},i}^* \sqrt{\beta_{\text{B2I2U},mk}}  {\bf a}^H_{\text{B2I},n}  {\bf a}_{\text{B2I},i} \sqrt{\beta_{\text{B2I2U},nk}}
\nonumber \\
&\times  \sum\limits_{s=1}^{M}  \sum\limits_{l=1}^{M}
 \mathbb{E}\left\{\varrho_{\text{I2U},mk,s} \varrho_{\text{I2U},nk,l}^* \right\}
{\left[ {\hat{\bar{\bf g}}}_{\text{I2U},mk} \right]}_s  {\left[{\hat{\bar{\bf g}}}_{\text{I2U},mm} \right]}_s^*
{\left[ {\hat{\bar{\bf g}}}_{\text{I2U},nk} \right]}_l^*  {\left[{\hat{\bar{\bf g}}}_{\text{I2U},nn} \right]}_l, \nonumber
\end{align}
where $\mathbb{E}\left\{\varrho_{\text{I2U},mk,s} \varrho_{\text{I2U},nk,l}^* \right\} = \zeta_{\text{y-I2U},mk,nk,sl}$  given by Lemma \ref{L2}.

Combining 1) and 2) yields
\begin{align}
   & C_{k,i}=\frac{\eta_i \rho_d}{N} \sum\limits_{m=1}^{K} \left(M \frac{\beta_{\text{B2I2U},mk} }{v_{\text{I2U},mk}} \left|{\bf a}^T_{\text{B2I},m}{\bf a}_{\text{B2I},i}^* \right|^2+
   MN\frac{\beta_{\text{B2I2U},mk}}{v_{\text{B2I},m}v_{\text{I2U},mk} }+
   +MN \frac{\beta_{\text{B2I2U},mk}}{v_{\text{B2I},m}}\right)
   \\
  &+\frac{\eta_i \rho_d}{N}\sum\limits_{m=1}^{K}
  \sum\limits_{s=1}^{M}   \sum\limits_{l=1}^{M}
  \left| {\bf a}^T_{\text{B2I},m} {\bf a}_{\text{B2I},i}^* \right|^2
 \zeta_{\text{y-I2U},mk,sl}
{\left[ {\hat{\bar{\bf g}}}_{\text{I2U},mk} \right]}_s  {\left[{\hat{\bar{\bf g}}}_{\text{I2U},mm} \right]}_s^*
{\left[ {\hat{\bar{\bf g}}}_{\text{I2U},mk} \right]}_l^*  {\left[{\hat{\bar{\bf g}}}_{\text{I2U},mm} \right]}_l
\nonumber \\
&+\frac{\eta_i \rho_d}{N}\sum\limits_{m=1}^{K} \sum\limits_{n\ne m}^{K}
   \sum\limits_{s=1}^{M}   \sum\limits_{l=1}^{M}
   {\bf a}^T_{\text{B2I},m} {\bf a}_{\text{B2I},i}^* \sqrt{\beta_{\text{B2I2U},mk}}  {\bf a}^H_{\text{B2I},n}  {\bf a}_{\text{B2I},i} \sqrt{\beta_{\text{B2I2U},nk}}
\nonumber \\
&\times
 \zeta_{\text{y-I2U},mk,nk,sl}
{\left[ {\hat{\bar{\bf g}}}_{\text{I2U},mk} \right]}_s  {\left[{\hat{\bar{\bf g}}}_{\text{I2U},mm} \right]}_s^*
{\left[ {\hat{\bar{\bf g}}}_{\text{I2U},nk} \right]}_l^*  {\left[{\hat{\bar{\bf g}}}_{\text{I2U},nn} \right]}_l. \nonumber
\end{align}

\subsection{Calculate $B_k$} \label{AC}
We rewrite $B_k$ as
\begin{align}
  B_k=\mathbb{E} \left\{ {\left| {{{\bf g}_{k}^T   {\bf w}_k }} \right|}^2 \right\}  -A_k.
\end{align}

Following the similar process of the calculation of $C_{k,i}$, we have

\begin{align}
& \mathbb{E} \left\{ {\left| {{{\bf g}_{k}^T   {\bf w}_k }} \right|}^2 \right\}
\\
&=
 \frac{\eta_k \rho_d}{N} \sum\limits_{m=1}^{K} \left(M \frac{\beta_{\text{B2I2U},mk} }{v_{\text{I2U},mk}} \left|{\bf a}^T_{\text{B2I},m}{\bf a}_{\text{B2I},k}^* \right|^2+
   MN\frac{\beta_{\text{B2I2U},mk}}{v_{\text{B2I},m}v_{\text{I2U},mk} }
   +MN \frac{\beta_{\text{B2I2U},mk}}{v_{\text{B2I},m}}\right)\nonumber
   \\
  &+\frac{\eta_k \rho_d}{N}\sum\limits_{m=1}^{K}
  \sum\limits_{s=1}^{M}   \sum\limits_{l=1}^{M}
  \left| {\bf a}^T_{\text{B2I},m} {\bf a}_{\text{B2I},k}^* \right|^2
 \zeta_{\text{y-I2U},mk,sl}
{\left[ {\hat{\bar{\bf g}}}_{\text{I2U},mk} \right]}_s  {\left[{\hat{\bar{\bf g}}}_{\text{I2U},mm} \right]}_s^*
{\left[ {\hat{\bar{\bf g}}}_{\text{I2U},mk} \right]}_l^*  {\left[{\hat{\bar{\bf g}}}_{\text{I2U},mm} \right]}_l
\nonumber \\
&+\frac{\eta_k \rho_d}{N}\sum\limits_{m=1}^{K} \sum\limits_{n\ne m}^{K}
   \sum\limits_{s=1}^{M}   \sum\limits_{l=1}^{M}
   {\bf a}^T_{\text{B2I},m} {\bf a}_{\text{B2I},k}^* \sqrt{\beta_{\text{B2I2U},mk}}  {\bf a}^H_{\text{B2I},n}  {\bf a}_{\text{B2I},k} \sqrt{\beta_{\text{B2I2U},nk}}
\nonumber \\
&\times
\zeta_{\text{y-I2U},mk,nk,sl}
{\left[ {\hat{\bar{\bf g}}}_{\text{I2U},mk} \right]}_s  {\left[{\hat{\bar{\bf g}}}_{\text{I2U},mm} \right]}_s
{\left[ {\hat{\bar{\bf g}}}_{\text{I2U},nk} \right]}_l^*  {\left[{\hat{\bar{\bf g}}}_{\text{I2U},nn} \right]}_l^*. \nonumber
\end{align}

As such, we have $B_k=C_{k,k}-A_k$.

Combining \ref{AA}, \ref{AB} and \ref{AC} yields the desired result.

\end{appendices}


\begin{thebibliography}{10}
\bibitem{tang2019mimo}
W. Tang, J. Y. Dai, et al, ``MIMO transmission through
reconfigurable intelligent surface: System design, analysis, and implementation,'' [online]. Available:https://arxiv.org/abs/1912.09955.

\bibitem{wu2018intelligent}
Q. Wu and R. Zhang, ``Intelligent reflecting surface enhanced wireless network: Joint active and passive beamforming design,'' in {\em 2018 IEEE Global Communications Conference (GLOBECOM)}, Abu Dhabi, United Arab Emirates, 2018, pp. 1-6.


\bibitem{wu2019intelligent}
Q. Wu and R. Zhang, ``Intelligent reflecting surface enhanced wireless network via joint active and passive beamforming,'' {\em  IEEE Trans. Commun.}, vol. 18, no. 11, pp. 5394--5409, 2019.

\bibitem{Yu2019MISO}
X. Yu, D. Xu and R. Schober, ``MISO wireless communication systems via intelligent reflecting surfaces," in {\em 2019 IEEE/CIC International Conference on Communications in China (ICCC)}, Changchun, China, 2019, pp. 735--740.

\bibitem{J.Gao}
J. Gao, C. Zhong, X. Chen, H. Lin, and Z. Zhang, ``Unsupervised learning for passive beamforming,'' {\em IEEE Wireless Commun. Lett.}, vol. 24, no. 5, pp. 1052-1056, May 2020.

\bibitem{X.Hu}
X. Hu, C. Zhong, Y. Zhu, X. Chen, and Z. Zhang, ``Programmable metasurface based multicast systems: Design and analysis,'' {\em IEEE J. Selected Areas Commun.}, vol. 38, no. 8, pp. 1763-1776, Aug. 2020.


\bibitem{huang2019reconfigurable}
C. Huang, A. Zappone, G. C. Alexandropoulos, M. Debbah, and C. Yuen, ``Reconfigurable intelligent surfaces for energy
efficiency in wireless communication,'' {\em IEEE Trans. Wireless Commun.}, vol. 18, no. 8, pp. 4157--4170, 2019.

\bibitem{wu2019beamforming_TCOM}
Q. Wu and R. Zhang, ``Beamforming optimization for wireless network aided by intelligent reflecting surface with discrete phase shifts,''  {\em IEEE Trans. Commun.}, vol. 68, no. 3, pp. 1838--1851, March. 2020.

\bibitem{wu2019beamforming}
Q. Wu and R. Zhang, ``Beamforming optimization for intelligent reflecting surface with discrete phase shifts,'' in {\em ICASSP 2019-2019 IEEE International Conference on Acoustics, Speech and Signal Processing (ICASSP)}, Brighton, United Kingdom, 2019, pp. 7830--7833.

\bibitem{li2019joint}
X. Li, J. Fang, F. Gao, and H. Li, ``Joint active and passive beamforming for intelligent reflecting surface-assisted massive
MISO systems,'' [online]. Available:https://arxiv.org/abs/1912.00728.

\bibitem{yuan2019intelligent}
J. Yuan, Y.-C. Liang, J. Joung, G. Feng, and E. G. Larsson, ``Intelligent reflecting surface-assisted cognitive radio system,'' [online]. Available:https://arxiv.org/abs/1912.10678.

\bibitem{yang2019intelligent}
Y. Yang, B. Zheng, S. Zhang and R. Zhang, ``Intelligent reflecting surface meets OFDM: Protocol design and rate maximization,'' {\em IEEE Trans. Commun.}, Early Access. DOI: 10.1109/TCOMM.2020.2981458.

\bibitem{fu2019intelligent}
M. Fu, Y. Zhou and Y. Shi, ``Intelligent reflecting surface for downlink non-orthogonal multiple access networks,'' in {\em 2019 IEEE Globecom Workshops (GC Wkshps)}, Waikoloa, HI, USA, 2019, pp. 1-6.

\bibitem{mu2019exploiting}
X. Mu, Y. Liu, L. Guo, J. Lin, and N. Al-Dhahir, ``Exploiting intelligent reflecting surfaces in multi-antenna aided NOMA
systems,'' [online]. Available:https://arxiv.org/abs/1910.13636.

\bibitem{ding2019simple}
Z. Ding and H. V. Poor, ``A simple design of IRS-NOMA transmission,'' {\em IEEE Commun. Lett.}, Early Access. Doi: 10.1109/LCOMM.2020.2974196.

\bibitem{Y.Zhang}
Y. Zhang, C. Zhong, Z. Zhang, and W. Lu, ``Sum rate optimization for two way communications with intelligent reflecting surface,'' {\em IEEE Commun. Lett.}, vol. 25, no. 5, pp. 1090-1094, May 2020.

\bibitem{wang2019joint}
P. Wang, J. Fang, and H. Li, ``Joint beamforming for intelligent reflecting surface-assisted millimeter wave communications,'' [online]. Available:https://arxiv.org/abs/1910.08541.

\bibitem{wang2019intelligent}
P. Wang, J. Fang, X. Yuan, Z. Chen, H. Duan, and H. Li, ``Intelligent reflecting surface-assisted millimeter wave
communications: Joint active and passive precoding design,'' [online]. Available:https://arxiv.org/abs/1908.10734.

\bibitem{ning2019channel}
B. Ning, Z. Chen, W. Chen, and Y. Du, ``Channel estimation and transmission for intelligent reflecting surface assisted THz
communications,'' [online]. Available:https://arxiv.org/abs/1911.04719.

\bibitem{cui2019secure}
M. Cui, G. Zhang and R. Zhang, ``Secure wireless communication via intelligent reflecting surface,''  {\em IEEE   Wireless Commun. Lett.}, vol. 8, no. 5, pp. 1410--1414, Oct. 2019.



\bibitem{Chen2019Intelligent}
J. Chen, Y.-C. Liang, Y. Pei, and H. Guo, ``Intelligent reflecting surface: a programmable wireless environment for physical
layer security,'' {\em IEEE Access}, vol. 7, pp. 82599–82612, 2019.

\bibitem{yu2019enabling}


X. Yu, D. Xu and R. Schober, ``Enabling secure wireless communications via intelligent reflecting surfaces,'' in {\em 2019 IEEE Global Communications Conference (GLOBECOM)}, Waikoloa, HI, USA, 2019, pp. 1--6.


\bibitem{wu2019joint}
Q. Wu and R. Zhang, ``Joint active and passive beamforming optimization for intelligent reflecting surface assisted SWIPT
under QoS constraints,'' [online]. Available:https://arxiv.org/abs/1910.06220.


\bibitem{pan2019intelligent}
C. Pan, H. Ren, K. Wang, M. Elkashlan, A. Nallanathan, J. Wang, and L. Hanzo, ``Intelligent reflecting surface enhanced MIMO broadcasting for simultaneous wireless information and power transfer,''[online]. Available:https://arxiv.org/abs/1908.04863.


\bibitem{mishra2019channel}
D. Mishra and H. Johansson, ``Channel estimation and low-complexity beamforming design for passive intelligent surface assisted MISO wireless energy transfer,'' in {\em ICASSP 2019-2019 IEEE International Conference on Acoustics, Speech and Signal Processing (ICASSP)}, Brighton, United Kingdom, 2019, pp. 4659--4663.



\bibitem{zheng2019intelligent}
B. Zheng and R. Zhang, ``Intelligent reflecting surface-enhanced OFDM: Channel estimation and reflection optimization,'' {\em IEEE Wireless Commun. Lett.}, vol. 9, no. 4, pp. 518--522, April. 2020.



\bibitem{he2019cascaded}
Z. He and X. Yuan, ``Cascaded channel estimation for large intelligent metasurface assisted massive MIMO,'' {\em IEEE Wireless Commun. Lett.}, vol. 9, no. 2, pp. 210--214, Feb. 2020.


\bibitem{wang2019channel}
Z. Wang, L. Liu, and S. Cui, ``Channel estimation for intelligent reflecting surface assisted multiuser communications,'' [online]. Available:https://arxiv.org/abs/1911.03084.

\bibitem{you2019intelligent}
C. You, B. Zheng, and R. Zhang, ``Intelligent reflecting surface with discrete phase shifts: Channel estimation and passive
beamforming,''  [online]. Available:https://arxiv.org/abs/1911.03916.

\bibitem{abeywickrama2019intelligent}
S. Abeywickrama, R. Zhang, and C. Yuen, ``Intelligent reflecting surface: Practical phase shift model and beamforming
optimization,'' [online]. Available:https://arxiv.org/abs/2002.10112.

\bibitem{han2019large}
Y. Han, W. Tang, S. Jin, C.-K. Wen, and X. Ma, ``Large intelligent surface-assisted wireless communication exploiting
statistical CSI,'' {\em IEEE Trans. Veh. Technol.}, vol. 68, no. 8, pp. 8238–8242, 2019.
\bibitem{X.Hu2}
X. Hu, J. Wang and C. Zhong, ``Statistical CSI based design for intelligent reflecting surface assisted MISO systems,'' accepted to appear in {\em Science China: Information Sciences}, 2020.

\bibitem{he2019adaptive}
J. He, H. Wymeersch, T. Sanguanpuak, O. Silv{\'e}n, and M. Juntti, ``Adaptive beamforming design for mmWave RIS-aided joint
localization and communication,'' [online]. Available:https://arxiv.org/abs/1911.02813.

\bibitem{marzetta2016fundamentals}
T. L. Marzetta, {\em Fundamentals of massive MIMO}. Cambridge University Press, 2016.

\bibitem{J.Zhang}
J. Zhang, Y. Zhang, C. Zhong, and Z. Zhang, ``Robust design for intelligent reflecting surfaces assisted MISO systems,'' accepted to appear in {\em IEEE Commun. Lett.}, 2020.

\bibitem{X.Hu1}
X. Hu, C. Zhong, M. Alouini, and Z. Zhang, ``Robust design for IRS-aided communication systems with user location uncertainty,'' submitted to IEEE Wireless Commun. Lett., 2020.



\end{thebibliography}

\nocite{*}
\bibliographystyle{IEEE}
 \begin{footnotesize}

\end{footnotesize}

\end{document}